\documentclass[aps,prx,twocolumn,amsmath,nofootinbib,longbibliography,amssymb,superscriptaddress,10pt]{revtex4-2}

\usepackage[english]{babel}
\usepackage{etoolbox}
\usepackage[shortlabels]{enumitem}
\usepackage{amsmath,amsthm,amssymb,amsfonts}
\usepackage{stmaryrd}
\usepackage{accents}
\usepackage{bm}
\usepackage[mathscr]{euscript}
\usepackage[linkcolor=blue,colorlinks=true]{hyperref}
\usepackage{graphicx}
\usepackage[caption=false]{subfig}
\graphicspath{
    {figures/}
}

\newtheorem{theorem}{Theorem}[section]

\newtheorem{lemma}[theorem]{Lemma}

\docsvlist{A,B,C,D,E,F,G,H,I,J,K,L,M,N,O,P,Q,R,S,T,U,V,W,X,Y,Z}

\docsvlist{A,B,C,D,E,F,G,H,I,J,K,L,M,N,O,P,Q,R,S,T,U,V,W,X,Y,Z}

\docsvlist{n,e,E}

\def \bra {\langle}
\def \ket {\rangle}

\def \inplane {\kappa}
\def \vphi {\varphi} 

\def \rms {\varepsilon} 
\def \ts#1{\textsuperscript{#1}}

\renewcommand \r {\bm{\mathrm{r}}}

\makeatletter 
    
\renewcommand\onecolumngrid{
\do@columngrid{one}{\@ne}%
\def\set@footnotewidth{\onecolumngrid}
\def\footnoterule{\kern-6pt\hrule width 1.5in\kern6pt}%
}

\renewcommand\twocolumngrid{
        \def\footnoterule{
        \dimen@\skip\footins\divide\dimen@\thr@@
        \kern-\dimen@\hrule width.5in\kern\dimen@}
        \do@columngrid{mlt}{\tw@}
}%

\makeatother  

\begin{document}

\title{Exactly Solvable 
Model of Randomly Coupled Twisted Superconducting Bilayers}
\author{Andrew C. Yuan}
\affiliation{Department of Physics, Stanford University, Stanford, CA 94305, USA}

\date{\today}

\begin{abstract}
Motivated by recent experiments on twisted junctions of cuprate superconductors (SC), it was proposed \cite{yuan2023inhomogeneity} that at zero temperature, a random first order Josephson coupling $J_1(\r) \cos \phi$ generates an ``effective" global second order coupling, $J_2\cos(2\phi)$, with a sign that favors $\phi = \pm \pi/2$, i.e., spontaneous breaking of time reversal symmetry (TRS).
To obtain a more controlled understanding of the suggested ``disorder-induced-order" mechanism, we construct an exactly solvable lattice mean field model and prove that when the disorder-average $\bar{J}_1=0$, the model exhibits a TRS breaking phase for all temperatures below the SC transition, i.e., $T_c = T_{\mathrm{TRSB}}$, regardless of the specific form of disorder. 
In the presence of nonzero $\bar{J}_1\ne 0$, we show that the two transitions split linearly for small $\bar{J}_1 \ll \inplane$ (where $\inplane$ is the in-plane SC stiffness), and that $T_{\mathrm{TRSB}}$ vanishes for $\bar J_1>  J_c$ where $ J_c= \overline{J^2_1}/\inplane$ in the weak disorder limit.
\end{abstract}
\maketitle

\section{Introduction}
Twisting and stacking 2D quantum materials have proven to be powerful tools for creating new materials \cite{cao2018correlated,cao2018unconventional,yankowitz2019tuning} and probing their properties \cite{polshyn_large_2019,inbar2023quantum}.
In particular, the Josephson coupling between two $d$-wave superconductors has been studied as a function of the twist angle $\theta$ in an attempt to reveal their underlying pairing symmetry.
Early experiments in cuprate superconductors did not show the expected angle dependence of the Josephson coupling \cite{li1999bi}.
However, in a beautiful recent set of experiments \cite{zhao2021emergent}, an observed $|\cos (2\theta)|$ dependence of the Josephson coupling has been observed, potentially resolving the discrepancy.

Though still more recent experiments \cite{XueTwist} have reproduced earlier results (with the Josephson coupling having no angle dependence), we conjecture that this discrepancy has an extrinsic origin and focus on the experiment showing the expected angle dependence. 
In this case, at $\theta = 45^\circ$, the lowest-order Josephson coupling $J_1$ vanishes due to $d$-wave symmetry, but a substantial second-order coupling $J_2$ has been observed.
Notably, such a phenomenon was predicted and studied by proposing an 
intrinsic 2\ts{nd} order Joesphson coupling as the underlying mechanism \cite{can2021high}. 
Possible implications such as time-reversal symmetry breaking (TRSB), chiral superconductivity and gapped topological behavior near $\theta=45^\circ$ have aroused attention within the community \cite{can2021high,tummuru2022josephson,mercado2022high}.

Experimentally, the precise microscopic mechanism of the 2\ts{nd} order coupling and the sign of $J_2$ are still under investigation.
The extreme anisotropy of BSCCO (cuprates used in the experiments) 
implies that the expected magnitude of the intrinsic $J_2$ would be too small to explain observations \cite{yuan2023inhomogeneity}. 
Therefore, an alternative mechanism involving locally symmetry-breaking (nematic) inhomogeneities was proposed \cite{yuan2023inhomogeneity}.
The local nematicity induces a spatially varying 1\ts{st} order Josephson coupling, which can generate an effective $J_2$ with the required sign to favor TRSB, i.e., phase difference $\phi=\pm \pi/2$ between the two twisted superconductors.
The magnitude of the effective $J_2$ is enhanced by a factor of $\xi/\xi_\text{sc}$ where $\xi$ is the disorder correlation length and $\xi_\text{sc}$ is the superconducting coherence length, and thus for reasonably large $\xi$, the proposed mechanism is sufficiently large to explain the experimental observations.

It is worth mentioning that the proposal in Ref. \cite{yuan2023inhomogeneity} has a few restrictions due to the nature of the implemented perturbative approach. (1). It focuses on the zero temperature scenario where the in-plane nearest neighbor interaction can be approximated by the standard nonlinear sigma model (NLSM).
Therefore, whether the TRSB behavior can be extended to finite temperatures, and more importantly, up to temperatures comparable to the SC transition $T_c$ (to achieve high $T_\mathrm{TRSB} \approx T_c$ topological superconductivity), remains an open question.
(2). The approach is only valid in the limit of weak (Gaussian) disorder. 
While this aligns with the physical conditions typically encountered in cases like twisted cuprates, it doesn't address the validity of similar conclusions in scenarios involving arbitrarily large (possibly non-Gaussian) disorder. 
Notably, the connection between the strength of disorder couplings and the magnitude of the SC gap in the TRSB phase suggests that robust topological superconductivity might necessitate substantial disorder interactions.
(3). In the extreme limit in which each twisted superconductor is a single 2D layer, the perturbative approach 
also encounters difficulties due to a 
logarithmic divergence in the computation of the "effective" $J_2$\footnote{ 
This infrared divergence was treated using a variation mean field approach, in which an emergent length scale served as a physical cut-off.}.
Whether this affects the claim of disorder induced TRSB behavior also is yet to be determined.

In this paper, we will provide an alternative approach, and solve exactly a model with all-to-all interactions \cite{velenik} for which mean-field theory is exact.
The constructed model permits us to circumvent (ill-controlled) perturbative techniques and answer many of the open questions discussed previously, albeit we pay the price of a less physical model.
Nevertheless, we prove that at twist angle $\theta_c =45^\circ$ so that $\bar{J}_1=  0$, the model exhibits TRSB behavior for all temperatures below the superconducting transition, i.e., $T_c = T_\mathrm{TRSB}$\footnote{The absence of vestigial TRSB above $T_c$ will turn out to be trivial within the context of mean-field theory}, regardless of the specific form and coupling strength of the disorder\footnote{Similar exact statements can be made to the non-disordered inherent $J_2$ model with all-to-all interactions. In fact, the proof will be much simpler and we suggest the reader understand the non-disordered problem before moving on to the disordered problem. See Appendix \eqref{app:MFT-rigor} (or more specifically, Theorem \eqref{app-thm:non-disordered}) for details.}.
Near twist angle $\theta \approx 45^\circ$ so that the disorder-average $\bar{J}_1 \ne 0$, we show that the two transitions split linearly for small $\bar{J}_1 \ll \inplane$ (where $\inplane$ is the in-plance SC stiffness), with the TRSB phase eventually disappearing when $\bar{J}_1$ reaches $\bar{J}_1 = \overline{J^2_1}/\inplane$ in the limit of weak disorder $\overline{J^2_1} \to 0$.
In the case of twisted BSCCO, due to the large anisotropy, this implies that topological superconductivity can only be achieved in a narrow region around the critical twist angle $\theta_c = 45^\circ$.

The paper is organized as follows. Sec. \eqref{sec:model} introduces the classical statistical all-to-all model under investigation as well as the motivation of its construction and its relation to the standard nearest-neighbor XY model. 
Sec. \eqref{sec:free-energy} computes the exact free energy of the all-to-all model and disucsses its inherent symmetries. 
The correspondence between spontaneous symmetry breaking (SSB) and the respective transition temperatures is made explicit.
Sec. \eqref{sec:pert-results} illustrates the perturbative results, which are well-controlled in given limits and independent of the specific form of the disorder.
They are then compared with the corresponding numerical results shown in in Fig. \ref{fig:mft}.
Finally, Sec. \eqref{sec:general-T} proves the main statement, i.e., Theorem \eqref{claim:zero-Jbar}. 
The paper will only outline and cite the necessary theorems, while leaving the proof in the Appendix.

\section{All-to-All Model: Motivation and Setup}
\label{sec:model}

We will represent the twist junction at twist-angle $\theta$ by a model Hamiltonian consisting of  two classical XY ferromagnets representing the phase degrees of freedom of the  individual superconducting planes with random-in sign local couplings between the two:
\begin{align}
    \label{eq:H}
    \mathcal{H}  
    &= -\inplane \sum_{\langle \r\r' \rangle ,l} \sigma_l(\r)\cdot \sigma_l(\r')- \sum_{\r} J(\r) \sigma_1(\r) \cdot \sigma_2(\r)
\end{align}
where the system is defined on a square lattice, $\sigma_l(\r) \in \mathbb{S}^1$ corresponds to the superconducting phase $\phi_l(\r)$ at site $\r$ and layer $l =\pm$, and the last term involving $\sigma_1(\r)\cdot \sigma_2(\r)=\cos \phi(\r)$ (where $\phi=\phi_1-\phi_2$) represents the first order Josephson coupling between the 2 layers. 
The disorder realization  $\bm{J} = (J(\r): \r\in \mathbb{Z}^2)$ satisfy $J(\r) = \bar{J}_\theta + \delta J(\r)$, where $\bar{J}_\theta = \bar{J}_0 \cos(2\theta)$ and $\theta$ corresponds to the twist angle, and $\delta \bm{J} =(\delta J(\r):\r\in \mathbb{Z}^2)$ are independently identically distributed (iid) \textit{even}\footnote{$\delta J(\r)$ and $-\delta J(\r)$ have the same distribution} random variables\footnote{It should be emphasized that the disorder does not necessarily have to be of Gaussian form.} such that the disorder average (denoted by $\overline{\cdots}$) satisfies $\overline{\delta J(\r)} = 0$ and $\overline{(\delta J(\r))^2} = \rms^2$.
Note that the twist angle $\theta$ enters entirely through $\bar{J}_\theta$ (and thus will often be omitted).

Physically, the lattice spacing of our model corresponds to  the spatial correlation length $\xi$ of the disordered inter-layer Josephson coupling, whose length scale can be estimated from experimental data  \cite{yuan2023inhomogeneity}. 
$\inplane$ is the SC phase stiffness of an individual CuO$_2$ plane. 
Since the cuprate superconductors are highly anisotropic, both $\bar{J}$ and $\rms$ are much smaller than $\inplane$, with typical magnitudes discussed in Ref. \cite{yuan2023inhomogeneity}, though we will not make this restriction in our analysis unless otherwise stated.
This model resembles a random field XY model that was studied by Cardy and Ostlund \cite{PhysRevB.25.6899}, however it is distinct in the fact that the type of disorder we account for preserves time-reversal symmetry (TRS).

Due to the difficulty of treating the model introduced in Eq. \eqref{eq:H}, we shall instead consider a modified version of the model with all-to-all interactions in each plane for which many aspects of the system can be computed exactly from a mean-field approach.  
We thus define the ``mean-field Hamiltonian,''
\begin{equation}
    \label{eq:H-MF}
    \mathcal{H}_\text{MF}  = -\frac{d \inplane }{V} \sum_{\r,\r' ,l} \sigma_l(\r)\cdot \sigma_l(\r')- \sum_{\r} J(\r) \sigma_1(\r)\cdot \sigma_2(\r) 
\end{equation}
where, in common with the model introduced in Ref. \cite{velenik}, the nearest-neighbor in-plane interaction in Eq. \eqref{eq:H} is replaced by an infinite range coupling between in-plane spins, $d$ is the dimension of the in-plane system\footnote{In experiment, $d=2$, but we leave it as a general parameter since our mean-field calculations hold for all $d$.}, and $V$ is the volume (number of lattice sites in a single layer).
It should be emphasized, however, that the inter-layer interaction is kept local in space.
\subsection{Relation of the Two Models}
Consider the nearest-neighbor model in Eq. \eqref{eq:H} and notice that the in-plane nearest-neighbor interaction between spins $\sigma \in \mathbb{S}^1$ can be rewritten as
\begin{align}
    \inplane \sum_{\bra \r \r'\ket}  \sigma_l (\r) \cdot \sigma_l (\r') &= \frac{\kappa}{2} \sum_{\r} \sigma_l(\r) \cdot \sum_{\r':\bra \r \r'\ket} \sigma_l(\r') \\
    &= \kappa d \sum_{\r} \sigma_l(\r) \cdot \left[ \frac{1}{2d} \sum_{\r':\bra \r \r'\ket} \sigma_l(\r') \right]
\end{align}
For every lattice site $\r$, its $2d$ nearest-neighbor $\r'$ spins generates an effective magnetic field which interacts with $\sigma (\r)$. The mean field approximation we consider here is then replacing the nearest-neighbor field with a system-average magnetic field, i.e.,
\begin{align}
    \frac{1}{2d} \sum_{\r':\bra \r \r'\ket} \sigma_l(\r')  \mapsto \frac{1}{V} \sum_{\r'} \sigma_l(\r')
\end{align}
Where $V$ is the system size, i.e., number of lattice sites. We thus obtain $\mathcal{H}_\text{MF}$ in Eq. \eqref{eq:H-MF}.

It should be emphasized that the two models (Eq. \eqref{eq:H} and \eqref{eq:H-MF}) are only expected to coincide asymptotically when the effective number of neighboring spins is large as occurs in long-range interactions or in high dimensions ($d\to \infty$)\footnote{A similar mean-field model (called the Curie-Weiss model) was constructed for the standard Ising model, and it was shown that the free energy density of the two models coincides asymptotically in the limit $d\to \infty$ \cite{thompson1974ising}.}, though, similar to any other mean-field theory, it is generally believed that the two share the same physics in $d\ge 4$ dimensions \cite{velenik}.
In lower dimensions $d=2,3$, fluctuations due to short-range interactions of Eq. \eqref{eq:H} may be important and we comment more on this at the end of this paper.

\section{Exact Free Energy in Thermodynamic Limit}
\label{sec:free-energy}
Using a Hubbard-Stratonovich transformation (see Appendix \eqref{app:MFT-HS}), we can obtain the free energy density of $\mathcal{H}_\text{MF}$ as
\begin{equation}
    \label{eq:freee-energy-MF}
    \mathcal{F}_\text{MF}=\lim_{V\to \infty} \overline{\frac{-1}{\beta V} \log \mathcal{Z}_{V,\bm{J}}} = \min_{w_{\pm} \in \mathbb{R}^2} F(w_+,w_-)
\end{equation}
Where the minimizers $w_l^\star$ (the star $^\star$ is to distinguish from arbitrary values $w_l$) correspond to the magnetizations\footnote{The actual magnetizations of each layer have a 1-1 correspondence with the minimizers $w_l^\star$ (see Appendix \eqref{app:MFT-HS}). A nontrivial magnetization in layer $l$ implies $w_l^\star \ne 0$ and vice-versa.} of each layer $l =\pm $, and
\begin{align}
    \label{eq:free-energy}
    F  &= \frac{1}{4 d\inplane} (w_+^2 +w_-^2) - \frac{1}{\beta}\overline{G_J} \\
    G_J &= \ln \iint d\sigma_+d\sigma_- \; e^{-\beta H_J},  \\
    \label{eq:two-site-H}
    -H_J &= w_+\cdot \sigma_+ +w_-\cdot \sigma_- + J \sigma_+\cdot \sigma_-
\end{align}
Where $H_J$ denotes a two-site Hamiltonian\footnote{Notice that $F$ can be regarded as the disordered-average free energy of $H_{J} +(w_+^2+w_-^2)/4d\inplane$} with spins $\sigma_\pm \in \mathbb{S}^1$. To simplify notation, let us define $\vphi$ as the phase difference between $w_\pm \in \mathbb{R}^2$ and
\begin{equation}
    \label{eq:new-variables}
    a^2 = \frac{1}{4} (w_+^2 +w_-^2),\quad  \eta = \frac{2|w_+||w_-|}{w_+^2 +w_-^2}
\end{equation}
Where $0 \le \eta \le 1$ is dimensionless and denotes the relative magnitudes between the two layer\footnote{$\eta^\star=1$ implies that the two layers have the same magnitude in magnetization $|w_+^\star|=|w_-^\star|$, while $\eta^\star=0$ implies that only one of the layers is magnetized, e.g., $|w_+^\star|>|w_-^\star|=0$.}. Then (see Appendix \eqref{app:MFT-rewrite})
\begin{align}
    \label{eq:free-energy-rewrite}
    F &=  \frac{1}{d\inplane}a^2 - \frac{1}{\beta}\overline{G_J}\\
    \label{eq:free-energy-phase-diff-rewrite}
    G_J &= \ln \int  I_0(2\beta a \sqrt{1+\eta \cos \phi} ) e^{\beta J \cos(\phi+\vphi)}d\phi
\end{align}
Where $I_0$ denotes the modified Bessel function. 

\subsection{Symmetries}
Notice that $F$ is independent of the average phase of $w_{\pm} \in \mathbb{R}^2$ and invariant under $\vphi \mapsto -\vphi$.
Hence, $F$ has an innate $U(1)\times \mathbb{Z}_2$ symmetry. 
A nonzero minimizer $a^\star >0$ denotes a spontaneous symmetry breaking (SSB) in $U(1)$ and corresponds to the SC transition $T_{U(1)} \equiv T_c$.
Similarly, a nonzero $\vphi^\star \ne 0$ denotes SSB in $\mathbb{Z}_2$ and corresponds to the TRSB transition, i.e., $T_{\dZ_2} \equiv T_{\mathrm{TRSB}}$. 
Since the phase difference $\vphi$ is only well-defined when $a>0$, it follows that the critical temperature corresponding to each SSB satisfies $T_{U(1)} \ge T_{\mathbb{Z}_2}$ within mean-field theory. 
It is also worth mentioning that $F$ has an additional $\mathbb{Z}_2$ symmetry regarding exchanging the layers $w_1 \leftrightarrow w_2$, and if broken, would correspond to a minimizer with $\eta^\star \ne 1$.
However, we will show that $\eta^\star =1$ always holds (mathematically, for $\bar{J}=0$ and numerically for $\bar{J}\ne 0$). 
Hence, there is no confusion as to which order parameter $T_{\mathbb{Z}_2}$ refers to.

\section{Perturbative Results}
\label{sec:pert-results}
\begin{figure}[ht]
\label{fig:mft}
\subfloat[\label{fig:mft-disorder}]{%
  \centering
  \includegraphics[width=0.8\columnwidth]{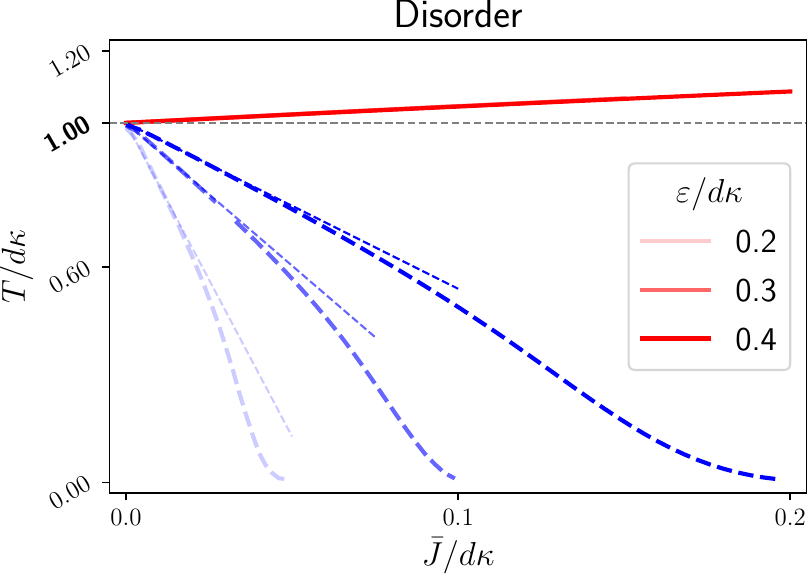}
}
\\
\subfloat[\label{fig:mft-T0}]{%
  \centering
  \includegraphics[width=0.8\columnwidth]{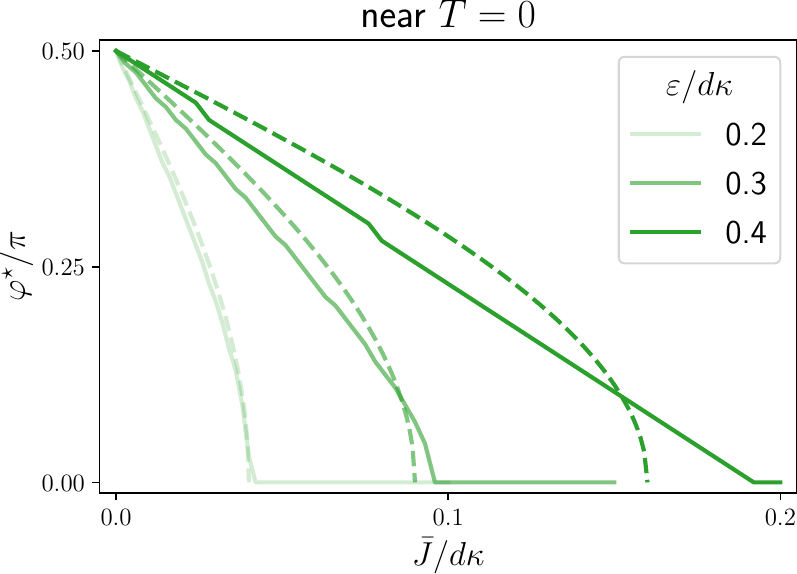}
}
\caption{Mean-Field. The numerical results assumes that the disorder $\delta J$ is of Gaussian distribution. 
This is the only time we assume an explicit distribution of $\delta J$.
(a). The solid/dashed lines represent the BKT/Ising transition temperatures $T_{U(1)},T_{\mathbb{Z}_2}$. The shorter dashed lines denote the linear approximation of $T_{\mathbb{Z}_2}$ with respect to $\bar{J}=\bar{J}_\theta$.
(b). The solid line denotes the minimizing $\vphi^\star$ computed at $\beta d\inplane =100$, while the dashed lines denote the asymptotic limit as $\rms \to 0^+$ given by Eq. \eqref{eq:phi-T0}.}
\end{figure}
\subsection{Near $T_{U(1)}$}
Using $a$ as a small parameter, $F$ can be expanded near $T_{U(1)}$ so that
\begin{align}
     F &= -c_2 a^2 + c_4 \frac{1}{4} a^4 +O(a^6) \\
    c_2 &= \beta \left(1 +\eta  \overline{r_1} \cos \vphi\right) -\frac{1}{d\inplane}\\
    c_4 &=  \beta^3 \left[1-\frac{1}{2}\eta^2 +\eta^2 \overline{r_1^2} \right.\\
    &\quad \left.+2\eta \overline{r_1}\cos \vphi +\eta^2 \left(\overline{r_1^2} -\frac{\overline{r_2}}{2}\right) \cos 2\vphi\right]\nonumber
\end{align}
Where $r_\nu = I_\nu (\beta J)/I_0(\beta J)$ and $I_\nu$ are the modified Bessel functions.

\begin{enumerate}[(i)]
    \item If $\bar{J} >0$ so that $\overline{r_1} >0$, then the quadratic term is first term in the series expansion to contain dependence on $\vphi, \eta$ so that $F$ is minimized by $\eta^\star=1,\vphi^\star = 0$ near $T_{U(1)}$. 
    This indicates a split transition, i.e., $T_{U(1)} > T_{\mathbb{Z}_2}$. 
    \item Conversely, if $\bar{J}=0$ so that $\overline{r_1}=0$, the quartic terms must be considered so that $F$ is minimized\footnote{Notice that $I_1^2 \ge I_0 I_2$ \cite{thiruvenkatachar1951inequalities} and thus $r_1^2 > r_2/2$} by $\eta^\star=1,\vphi^\star = \pm \pi/2$ near $T_{U(1)}$. 
    Physically speaking, the disorder generates an ``effective" $\cos 2\vphi$ term so that $T_{U(1)} = T_{\mathbb{Z}_2}$\footnote{The series expansion indeed proves that TRSB behavior onsets precisely at the superconducting transition, provided that the transition is of second order, i.e., the minimzer $a^\star \to 0$ as $T\to T_{U(1)}$. However, it does not answer whether the TRSB behavior sustains at lower temperatures. Indeed, in general short-range lattice models (e.g., standard XY or Ising model), the Ginibre inequality \cite{ginibre1970general} guarantees that (ferromagnetic) ordering increases as the temperature decreases. However, due to the long-range interactions of the current model, other formal techniques must be developed, which are postponed to Theorem. \eqref{claim:zero-Jbar}.}.
\end{enumerate}

Similar to Ginzburg-Landau, the $U(1)$ transition occurs exactly when the quadratic coefficient $c_2$ transitions from a negative to a positive value, i.e., $c_2 =0$, and thus $T_{U(1)}$ is determined self-consistently via finding the fixed point of 
\begin{equation}
    \hat{T}= 1+\overline{r_1}
\end{equation} 
Where $\hat{\cdots} = \cdots/d\inplane$ is dimensionless\footnote{Notice that in the absence of disorder, $T_{U(1)}=d\inplane$ and thus warrants the definition of dimensionless quantities}.

Numerical calculations also suggest that $\eta^\star =1$ for general $\bar{J}>0$ and temperatures $T$, albeit we do not have a rigorous proof.
Assuming $\eta^\star= 1$, the 2$^\text{nd}$ transition $T_{\mathbb{Z}_2}$ can then be determined exactly via the self-consistently equations (see Appendix \eqref{app:MFT-self-consistency}), i.e.,
\begin{align}
    \label{eq:T-Z2}
    |\hat{w}_{\pm } | = 2\overline{\bra \cos \phi_1\ket_{J}}, \quad 
    \hat{T} = \overline{\bra (\sin \phi_+ -\sin \phi_-)^2\ket_J}
\end{align}
Where $\bra\cdots\ket_J$ denotes the thermal average with respect to the two-site Hamiltonian $H_J$ with $\vphi =0$. Fig. \ref{fig:mft-disorder} shows the 2 transitions (solid and dashed lines) $T_{U(1)},T_{\mathbb{Z}_2}$ as a function of the dimensionless quantities $\hat{\bar{J}}$ and $\hat{\rms}$.
The smaller dashed lines show the linear approximation of the transition temperatures near $\hat{T} \approx 1$ and with $\bar{J} / \rms \to 0$ (see Appendix \eqref{app:MFT-critical-T}), that is,
\begin{align}
    \hat{T}_{U(1)}&= 1+\frac{1}{2} \hat{\bar{J}} \left[ 1 -\frac{\hat{\rms}^2}{2^2} +O(\hat{\rms}^4) \right]+O(\hat{\bar{J}}^2) \\
    \hat{T}_{\mathbb{Z}_2}&=1-\hat{\bar{J}} \left[ \frac{2}{3} \frac{1}{\hat{\rms}^2} +O(1) \right] +O(\hat{\bar{J}}^2)
\end{align}
Notice that the slope of $\hat{T}_{\dZ_2}$ with respect to $\bar{J}$ is $\propto 1/\hat{\rms}^2$.
This implies that as the disorder strength $\hat{\rms} \to 0$, the region of $\mathbb{Z}_2$ SSB becomes extremely narrow in $\bar{J}$, despite $T_{\mathbb{Z}_2} = T_{U(1)}$ when $\bar{J}=0$.

\subsection{Near $T =0$ and small disorder $\rms \to 0^+$}
On a first glance of the two-site Hamiltonian $H_J$ in Eq. \eqref{eq:two-site-H}, one would naively attempt to obtain near $T=0$ results by using the nonlinear sigma model (NLSM)\footnote{Approximating $\cos \theta \approx 1-\theta^2 /2$ in the exponent to obtain a Gaussian measure, which in our case, is only valid when $\beta a \gg 1$. } with respect to spins $\sigma_{\pm} \in \mathbb{S}^1$.
However, due to disorder $J$, the expansion near $T=0$ will diverge as $a \to 0^+$ and thus cannot be readily minimized
(see Appendix \eqref{app:MFT-nls}). 
However, in the limit of small disorder $\rms \to 0^+$, we can safely assume that the minimizer $a^\star(\rms)$ is sufficiently close to the non-disordered minimizer $a^\star(\rms=0)\ne 0$, and thus the divergence of $F$ near $a\to 0^+$ does not affect us. 
Therefore, we find that (see Appendix \eqref{app:MFT-nls}) in the limit $\bar{J}/\rms \to 0$,
\begin{equation}
    \label{eq:phi-T0}
    \cos \vphi^\star = \hat{\bar{J}}/\hat{\rms}^2
\end{equation}
In particular, if $\hat{\bar{J}}\le  \hat{\rms}^2$, then the system exhibits $\mathbb{Z}_2$ SSB near $T=0$, where $\vphi^\star = \pm \pi/2$ when $\bar{J}=0$. Indeed, Fig. \ref{fig:mft-T0} compares the asymptotic limit (dashed lines) with the numerical calculations (solid) for different disorder strength $\rms$.

\section{General $T$}
\label{sec:general-T}

We emphasize that beyond special limits, non-perturbative techniques are essential in obtaining exact result and deeper insights to strong interactions.
Indeed, for general temperatures $T$ (and disorder strength $\rms$), one can prove rigorously that\footnote{We note that the formal techniques developed here can also be used to show that the analagous all-to-all model with an inherent non-disordered $J_2 \cos 2\phi$ inter-layer interaction has the same properties. In fact, the proof is much simpler than the disordered case and we suggest that the reader understand the non-disordered problem before moving on to the disordered problem. See Appendix \eqref{app:MFT-rigor} (or more specifically, Theorem \eqref{app-thm:non-disordered}) for details.}
\begin{theorem}[Orientation, see \eqref{app-thm:Z2-orient}]
    \label{claim:pos-Jbar} 
    If $\bar{J} \ge 0$, then $F$ is minimized when $|\vphi^\star| \le \pi/2$. Equivalently, $F$ is minimized when $\bar{J} \cos \vphi^\star \ge 0$.
\end{theorem}
\begin{theorem}[TRSB, see \eqref{app-thm:TRSB}]
    \label{claim:zero-Jbar} 
    If $\bar{J} = 0$, then $F$ is minimized via $\vphi^\star =\pm \pi/2$ and $\eta^\star =1$
\end{theorem}
Where the proof only relies on the probability distribution of $\delta J$ being even (see Appendix \eqref{app:MFT-rigor} for details).
Due to the definitions of $a,\vphi, \eta$, it should be noted that $\vphi,\eta$ are only well-defined when $a>0$, and thus the previous theorems implicitly assume that $a>0$.

Theorem \eqref{claim:pos-Jbar} is the formal statement which confirms the physical expectation that the system wishes to minimize the Josephson coupling in Eq. \eqref{eq:H-MF}.
More specifically, let us naively replace the disorder term in Eq. \eqref{eq:H-MF} with its average value $J(\r)\mapsto \bar{J}$ so that the average value generates a 1\ts{st} order Josephson $J_1 \sim \bar{J}$ coupling. 
On the other hand, the disorder $\delta J(\r)$ is expected to generate an ``effective" 2\ts{nd} order Josephson coupling term $J_2$ with a sign that favors a phase difference $\phi = \pm \pi/2$ (as described in previous work \cite{yuan2023inhomogeneity} near zero temperature).
In this naive setup, the interaction between the two layers is described by a competition between the 1\ts{st} and 2\ts{nd} order Josephson coupling, schematically written as $-\bar{J} \cos \phi + J_2\cos 2\phi$.
If $J_1 > 0$ ($<0$, respectively), then we would expect the ground state phase difference to be $|\phi|\le \pi/2$ ($\ge \pi/2$).

More importantly, theorem \eqref{claim:zero-Jbar} proves the existence of a TRSB phase generated by disorder, indicating that the subtle logarithm divergence seen in Ref. \cite{yuan2023inhomogeneity} may not affect the overall conclusion.
In fact, the TRSB phase extends up to the superconducting transition, which suggests the possibility of disorder induced high $T_c$ topological superconductivity.




\section{Summary and Discussion}

The paper presents a model that is exactly solvable for all-to-all interactions. 
Notably, it introduces rigorous methods to demonstrate the presence of a disorder-induced phase with broken time-reversal symmetry (TRSB)\footnote{
The physical interpretation of this result, as discussed previously in section \eqref{sec:general-T}, can be seen in the following manner. 
The disordered average $\bar{J}$ generates a 1\ts{st} order Josephson coupling $J_1\sim \bar{J}$, while the disorder term $\delta J(\r)$ is expected to generate an ``effective" 2\ts{nd} order Josephson coupling $J_2 \sim \rms^2$ with a sign that favors a phase difference $\phi = \pm \pi/2$ (as described in previous work \cite{yuan2023inhomogeneity} near zero temperature). 
Taken together, the interlayer interaction is schematically related to the non-disorder problem, i.e., $-J_1 \cos \phi +J_2 \cos 2\phi$.
Note that within the context of Ginzburg Landau theory, where $\psi_\pm \in \dC$ are the Landau order parameters for each layer $l=\pm$ \cite{can2021high}, only the quadratic terms determine the critical transition temperature. 
Since the $J_1$ term is quadratic (schematically written as $\psi_+^\dagger \psi_- +\text{h.c.}$) and the $J_2$ term is quartic (i.e., $(\psi_+^\dagger \psi_-)^2 +\text{h.c.}$), we see that the two transitions must coincide at $J_1=0$, and split when $J_1\ne 0$.
}.
When $\bar{J}=0$, we prove that the TRSB phase persists up to the superconducting transition temperature ($T_{\mathrm{TRSB}} = T_c$), independent of the specific disorder characteristics. 
The same techniques are applied to analyze the non-disordered (inherent $J_2$) problem with vanishing 1\ts{st} order Josephson coupling $J_1=0$ \cite{can2021high}, yielding analogous conclusions (see Theorem \eqref{app-thm:non-disordered}). 
Consequently, the developed framework could serve as a formal approach to validate results anticipated by Ginzburg-Landau theory in closely related scenarios, such as those referenced in \cite{bojesen2014phase, maccari2022effects}, and possibly in a more general setting as well.

We also note that while twisted bilayer superconductors \cite{XueTwist} were the motivation for the proposed exactly solvable model, similar constructions can potentially be applied in a boarder context, since the model is not concerned with the microscopic details.
Indeed, effective Josephson couplings incuded by frustration from three (or more) coupled bands \cite{lee2009pairing,agterberg1999conventional,ng2009broken,stanev2010three,bojesen2014phase,maccari2022effects}, inter-band scattering within a 2-band system \cite{bobkov2011time} or boundary effects \cite{bahcall1996boundary,rainer1998andreev} have been previously proposed as possible mechanisms of generating TRSB behavior.
Consequently, one would anticipate that disorder can induce similar phenomena within such proposals, as demonstrated in this study focusing on the bilayer system.

It should be noted that the mean-field model is only quantitatively reliable when the effective number of neighboring spins is large as occurs when there are long-range interactions or in high dimensions ($d\to \infty$) \cite{thompson1974ising}.
In low dimensions ($d=2,3$),  fluctuations due to short-range interactions may induce distinct statistical behavior, e.g., such as the possibility of a \textit{vestigial} TRSB phase, i.e., $T_{\mathrm{TRSB}} > T_c$ \cite{yuan2023fraun}.
Indeed, a similar possibility was suggested for the non-disorder inherent $J_2$ problem by applying an RG scheme \cite{zeng2021phase}, though further research (both numerical \cite{song2022phase} and mathematical \cite{yuan2023vestigial}) seems to suggest otherwise.

Another worry would be due to the Imry-Ma argument \cite{imry1975random}, which precludes long-range ordering of a continuous symmetry in the presence of disorder in low dimensions.
However, the Imry-Ma argument considers the scenario where the disorder has the same degree of freedom as the spins (e.g., if $\sigma \in \dS^1$, then so is the disorder field), while the disorder considered in this paper is ``effectively" within a hyper-plane of the continuous $U(1)$ symmetry, i.e., the disorder term $J(\r)$ is coupled to the phase difference $\phi$, and acts along the $x$-axis via the interaction $J(\r)\cos \phi$, while $\phi$ is in $\dS^1\cong U(1)$.
The difference can result in distinct statistical behavior as discussed in Ref. \cite{crawford2011random,crawford2013random}, in which it was proved mathematically that the system exhibits long-range order for sufficiently low temperatures and weak disorder.
This provides confidence that the certain aspects of the mean-field phenomena will extend to low dimensions (e.g., the existence of TRSB), circumventing the standard Imry-Ma argument.

\section{Acknowledegements}

I am grateful for Steve A. Kivelson's support and generosity during this project and also for providing extensive comments and suggestions on the draft.  This work was supported, in part, by NSF Grant No. DMR-2000987 at Stanford University. 
\bibliography{main.bbl}

\appendix
\onecolumngrid

\section{Hubbard-Stratonovich Transform and the Free Energy}
\label{app:MFT-HS}

In our main text, we assumed that $J(\r)$ are iid (independent identically distributed) even random variables with mean $\bar{J}$ and variance $\overline{(\delta J)^2}$.
However, since our results will not depend on the specific form of the distribution, in this section, it's instructive to consider the case where $\delta J$ is $\delta$-distributed, that is, $\delta J= \pm \rms$\footnote{
Notice that in the main text, we used $\rms^2$ to denote the disorder strength $\overline{(\delta J)^2}$. 
Therefore, if we wish to be consistent, we would require that $\delta J = \pm \rms/\sqrt{2}$. 
However, since this is a constant factor which will not affect our analysis, we shall adopt the simpler notation and omit the $1/\sqrt{2}$.} 
where $\rms >0$ and the $\pm$ sign is chosen randomly at each lattice site $\r$. 
The generalization to even probability distributions $\dP$ (i.e., $\dP[\delta J]=\dP[-\delta J]$ for all values of $\delta J$) is straightforward and we will comment on the extension in end of this section, i.e., Appendix \eqref{app:MFT-HS-extension}.
To simplify notation, we will also use $\mathbb{E}_{\delta J}[\cdots]$ to denote averaging over the disorder $\delta J$.

The Hubbard-Stratonovich (HS) transform is an exact transformation using the following fact \cite{velenik}
\begin{equation}
    \exp{(ax^2)} = \frac{1}{\sqrt{\pi a}} \int_\mathbb{R} \exp{\left(-\frac{y^2}{a} +2 xy\right)} dy
\end{equation}
In particular, we find that
\begin{equation}
    \exp \left(\frac{\beta d\inplane}{V} \left( \sum_{\r} \sigma_l (\r) \right)^2  \right) = \frac{\beta V}{4\pi d\inplane} \int_{\mathbb{R}^2} dw_l \exp \left[ -\frac{\beta V}{4 d\inplane} w_l^2 +\beta w_l \cdot \sum_{\r} \sigma_l(\r) \right]
\end{equation}
Notice that by using the HS transform, we obtain a term linear in $\sigma_l (\r)$, which implies that the exponential term can be decomposed into the product of local functions $\prod_{\r} f_l(\sigma_l(\r))$, each depending on the spin configuration $\bm{\sigma}_l \equiv (\sigma_l(\r):\r\in V)$ implicitly through $\sigma_l(\r)$ at lattice site $\r$. 
Since we intend to integrate over all spin configurations, i.e., $\sigma_l(\r) \in \dS^1$ for all $\r$, we can interchange the order of product and integration, i.e., schematically,
\begin{equation}
    \int_{(\dS^1)^V}d\bm{\sigma}_l \prod_{\r} f_l (\sigma_l (\r)) = \left[\int_{\dS^1} d\sigma_l f(\sigma_l)\right]^V, \quad d\bm{\sigma}_l \equiv \prod_{\r} d\sigma_l(\r)
\end{equation}
More specifically, for a given disorder realization $\bm{J}\equiv (J(\r):\r\in V)$,
\begin{align}
    \label{eq:partition-with-spin}
    Z_{V,\bm{J}} = \left( \frac{\beta V}{4\pi d\inplane}\right)^2 \iint_{\mathbb{R}^2\times \mathbb{R}^2} dw_\pm \exp \left[ -\frac{\beta V}{4 d\inplane} (w_+^2 +w_-^2)\right] \prod_{\r} \iint_{\mathbb{S}^1\times \mathbb{S}^1} d\sigma_\pm \exp \left[ \beta \sum_l w_l \cdot \sigma_l +J(\r) \sigma_+ \cdot \sigma_- \right]
\end{align}
By using the HS transform, the lattice site $\r$ dependence of the spins $\sigma_l$ inside the integral disappears and thus what remains is the possible $\r$ dependence of the inter-layer interaction $J(\r)$. 
Let $V_\pm$ denote the number of lattice sites $\r$ with $\delta J(\r) = \pm \rms$ so that $V_+ +V_-=V$ and let $\rho(\delta \bm{J})=V_+/V$ so that
\begin{align}
    \label{app-eq:delta-G}
    Z_{V,\bm{J}} &= \left( \frac{\beta V}{4\pi d\inplane}\right)^2 \iint_{\mathbb{R}^2\times \mathbb{R}^2} dw_\pm \exp \left[ -\frac{\beta V}{4 d\inplane} (w_+^2 +w_-^2) +V_+ G_{\bar{J}+\rms}(w_+,w_- ) +V_- G_{\bar{J}-\rms}(w_+,w_- ))\right] \\
    &=\left( \frac{\beta V}{4\pi d\inplane}\right)^2 \iint_{\mathbb{R}^2\times \mathbb{R}^2} dw_\pm \exp \left[ -\beta V \psi_{\bar{J},\rho(\delta \bm{J})}(w_+,w_-)\right]
\end{align}
Where $G_{J=\bar{J}\pm \rms}$ is defined using the two-site Hamiltonian in Eq. \eqref{eq:two-site-H} and
\begin{align}
    \psi_{\bar{J},\rho(\delta \bm{J})}(w_+,w_- ) = \frac{1}{4d\kappa} (w_+^2 +w_-^2) -\frac{1}{\beta} \left(\rho G_{\bar{J}+\rms} +(1-\rho) G_{\bar{J}-\rms} \right)
\end{align}
Therefore, we see that the partition function $Z_{V,\bm{J}}$ only depends on $\bm{J}=(J(\r):\r\in V)$ implicitly via the average value $\bar{J}$ and $\rho(\delta \bm{J})$ and thus warrants the notation $Z_{V,\bm{J}} = Z_{V,\bar{J},\rho(\delta \bm{J})}$. Notice that $\rho(\delta \bm{J})$ can be rewritten as
\begin{equation}
    \rho(\delta \bm{J}) = \frac{1}{V} \sum_{\r} \frac{\delta J(\r) +\rms}{2\rms}
\end{equation}
By the central limit theorem, $\rho (\delta J) \to \mathbb{E}_{\delta \bm{J}} [\rho (\delta \bm{J})] = 1/2$ as $V\to \infty$. 
Therefore, we will consider the partition function $Z_{V,\bar{J},\rho=1/2}$ defined with $\psi_{\bar{J},\rho=1/2}$, and show that the free energies of the 2 partition functions must converge in the thermodynamic limit $V\to \infty$. 
More specifically, we shall prove the following.

\begin{theorem}[Free Energy]
    \label{app-thm:free-energy}
    \begin{equation}
        \lim_{V\to\infty} \mathbb{E}_{\delta \bm{J}} \frac{-1}{\beta V} \log Z_{V,\bar{J},\rho(\delta \bm{J})} = \lim_{V\to\infty} \frac{-1}{\beta V} \log Z_{V,\bar{J},\rho=1/2} = \min_{w_\pm \in \mathbb{R}^2} \psi_{\bar{J},\rho=1/2}(w_+,w_-)
    \end{equation}
\end{theorem}
We leave the proof of this main statement to the end of this section. Instead, let us try to understand why the statement ``should" be true and its physical implications.
Indeed, when considering $Z_{V,\bar{J},\rho=1/2}$, due to the exponential weight $\propto e^{-\beta V \psi_{\bar{J},\rho=1/2}}$, at least schematically, only the minimizers $w_\pm^\star$ of $\psi_{\bar{J},\rho=1/2}$ should contribute to the free energy in the thermodynamic limit $V\to \infty$.
A similar logic applies when computing the average magnetization of each layer $l=\pm$, i.e., by taking the derivative of the weight $e^{-\beta V \psi_{V,\bar{J},\rho(\delta \bm{J})}}$ with respect to $w_l$, we obtain
\begin{equation}
    \left\bra \frac{1}{V} \sum_{\r} \sigma_l (\r)\right\ket_{V,\bm{J}}=\frac{1}{2d\inplane} \bra w_l\ket_{V\psi_{\bar{J},\rho(\delta \bm{J})}}
\end{equation}
Where the right-hand-side the average with respect to the measure $\propto e^{-\beta V \psi}$. In the thermodynamic limit $V\to\infty$, only the minimizer $w_\pm^\star$ will contribute to the average on the right-hand-side, and thus at least schematically, 
\begin{equation}
    \lim_{V\to\infty}\mathbb{E}_{\bm{J}}\left\bra \frac{1}{V} \sum_{\r} \sigma_l (\r)\right\ket_{V,\bm{J}}=\frac{1}{2d\inplane} w_l^\star
\end{equation}
Technically, due to exact $U(1)$ symmetry of the mean-field model, the average magnetization $\bra m\ket=0$. Rather, we should compute the average squared magnetization $\bra m^2\ket$ to account for $U(1)$ symmetry. However, the reasoning is the same and thus we see that there is a 1-1 correspondence between the minimizers $w_l^\star$ and the magnetization of each layer $l=\pm$, as we claimed in the main text.


The rest of this subsection will be contributed to proving Theorem \eqref{app-thm:free-energy}, where the proof involves 2 steps.
The first step is that the free energy corresponding to $Z_{V,\bar{J},\rho=1/2}$ converges and has the limit given on the right-hand-side (RHS). The second step is to show the first equality, i.e., the two partition functions have the same limit. We will also need the following lemma, obtained from Exercise 2.6 in Ref. \cite{velenik}, which we repeat here (albeit slightly modified for our case) for completeness.

\begin{lemma} [Exercise 2.6 in Ref. \cite{velenik}]
\label{app-lem:saddle-point}
Let $f:\mathbb{R} \to \mathbb{R}$ be analytic and bounded below by $O(x^2)$ for sufficiently large $x$. Then
\begin{equation}
    \lim_{N\to \infty}\sqrt{N} \int_{-\infty}^{\infty} e^{-N (f(x) - \min f(x)) dx} >0
\end{equation}
\end{lemma}
\begin{proof}
    Since this was left as an exercise in the reference, we shall provide a short proof here for completeness. Indeed, since $f(x)$ is analytic and nontrivial, we see that its minimums cannot have a limit. Since $f\to \infty$ at large $x$, we see that it can only have a finite number of minima. Therefore, without loss of generality, we shall assume that $f$ has a unique minimum at $x=0$ and take $f(0)=0$. Without loss of generality, we shall also assume that $f''(0) >0$, since if not, we can always take the lowest nonzero coefficient (which must be of even order of derivation since it is a min) in the Taylor series.
    By Taylor series, we see that for sufficiently small $\epsilon >0$, there exists $c(\epsilon),C(\epsilon)>0$ such that
    \begin{equation}
        c(\epsilon) \le \frac{f(x)}{x^2} \le C(\epsilon), \quad |x|\le \epsilon
    \end{equation}
    Notice that we can take $c(\epsilon), C(\epsilon) \to f''(0)/2$ as $\epsilon \to 0$. There also exists $R$ and $M$ such that if $|x| \ge R$, then $f(x) \ge Mx^2$. Hence, we see that
    \begin{align}
        \sqrt{N} \int e^{-N f(x)} dx &= \int e^{-N f(x/\sqrt{N})} dx\\
        &= \int_{|x| \le \epsilon \sqrt{N}}e^{-N f(x/\sqrt{N})} +\int_{\epsilon \sqrt{N} < |x| \le R\sqrt{N}} e^{-N f(x/\sqrt{N})} +\int_{|x| >R\sqrt{N}} e^{-N f(x/\sqrt{N})}
    \end{align}
    otice that
    \begin{equation}
        \int_{|x| > R\sqrt{N}} e^{-N f(x/\sqrt{N})} dx \le \int_{|x| > R\sqrt{N}} e^{-Mx^2} \to 0, \quad N\to \infty
    \end{equation}
    Since $f$ has a unique min valued $\min f = 0$ at $x=0$, we see that there exists $g(\epsilon)$ such that $f(x) \ge g(\epsilon) >0$ for $|x| >\epsilon$ and thus
    \begin{equation}
        \int_{\epsilon \sqrt{N} < |x| \le R\sqrt{N}} e^{-N f(x/\sqrt{N})} dx \le e^{-Ng(\epsilon)} \sqrt{N} \times (R-\epsilon) \to 0, \quad N\to \infty 
    \end{equation}
    Hence, only the first term will contribute to the limit. Indeed, notice that
    \begin{equation}
        \int_{|x| \le \epsilon \sqrt{N}}e^{-N f(x/\sqrt{N})} \le \int_{|x| \le \epsilon \sqrt{N}} e^{-c(\epsilon)x^2} dx \le \sqrt{\frac{\pi}{c(\epsilon)}}
    \end{equation}
    And similarly,
    \begin{align}
        \int_{|x| \le \epsilon \sqrt{N}}e^{-N f(x/\sqrt{N})} &\ge \int_{|x| \le \epsilon \sqrt{N}} e^{-C(\epsilon)x^2} dx \\
        \liminf_{N\to\infty} \int_{|x| \le \epsilon \sqrt{N}} e^{-N f(x/\sqrt{N})} &\ge \sqrt{\frac{\pi}{C(\epsilon)}}
    \end{align}
    Therefore, we have
    \begin{equation}
        \sqrt{\frac{\pi}{C(\epsilon)}} \le \liminf_{N\to\infty} \int e^{-N f(x/\sqrt{N})} \le \limsup_{N\to\infty} \int e^{-N f(x/\sqrt{N})} \le \sqrt{\frac{\pi}{c(\epsilon)}}
    \end{equation}
    Take $\epsilon \to 0$ and we see that the statement follows.
\end{proof}

\begin{proof}[Proof of Theorem \eqref{app-thm:free-energy}]
    Notice that the first step is nothing but a multi-dimensional application of the previous Lemma \eqref{app-lem:saddle-point}. Indeed, $F$ is analytic with respect to $|w|^2, \vphi, \eta$ and thus 
    \begin{equation}
        \lim_{V\to\infty} \frac{-1}{\beta V} \log Z_{V,\bar{J},\rho=1/2} = \min_{w_\pm \in \mathbb{R}^2} \psi_{\bar{J},\rho =1/2}(w_+,w_-)
    \end{equation}
    Let us now attempt to show the 1$^\text{st}$ equality. Indeed, notice that
    \begin{align}
        Z_{V,\bar{J},\rho(\delta J)} = Z_{V,\bar{J},\rho=1/2} \left\bra e^{-\beta V (\psi_{\bar{J},\rho(\delta J)} - \psi_{\bar{J},\rho=1/2}) } \right\ket_{V,\bar{J}, \rho=1/2}
    \end{align}
    Where the $\bra\cdots \ket$ is the average with respect to the weight $\propto e^{-\beta V \psi_{\bar{J},\rho=1/2} (w_+,w_-)}$, integrated over $w_\pm$, and normalized via $Z_{V,\bar{J},\rho=1/2}$. Therefore,
    \begin{align}
        \frac{1}{V} |\log Z_{V,\bar{J},\rho(\delta J)} -\log Z_{V,\bar{J},\rho=1/2}| \le \frac{1}{V} \left|\log \left\bra e^{-\beta V (\psi_{\bar{J},\rho(\delta J)} - \psi_{\bar{J},\rho=1/2})} \right\ket_{V,\bar{J}, \rho=1/2} \right|
    \end{align}
    Notice that there exists some constant $C>0$ such that
    \begin{equation}
        |\psi_{\bar{J},\rho(\delta J)} - \psi_{\bar{J},\rho=1/2}| \le \frac{C}{\beta} \left| \rho(\delta J) -\frac{1}{2}\right|
    \end{equation}
    Hence,
    \begin{equation}
         \frac{1}{V} |\log Z_{V,\bar{J},\rho(\delta J)} -\log Z_{V,\bar{J},\rho=1/2}| \le C \left| \rho(\delta J) -\frac{1}{2}\right|
    \end{equation}
    Therefore, we see that after disorder averaging, we have
    \begin{equation}
        \mathbb{E}_{\delta J} \frac{1}{V} |\log Z_{V,\bar{J},\rho(\delta J)} -\log Z_{V,\bar{J},\rho=1/2}| \le C \mathbb{E}_{\delta J}\left| \rho(\delta J) -\frac{1}{2}\right| \to 0, \quad V\to\infty
    \end{equation}
    Where the limit follows from the central limit theorem. Therefore, we see that the statement follows.
\end{proof}

\subsection{Extension to Arbitrary (Even) Distributions of $\delta J$}
\label{app:MFT-HS-extension}
We comment that when $\delta J$ is $\delta$-distributed, i.e., $\delta J =\pm \rms$, we see that $\psi_{\bar{J},\rho=1/2}$ in Theorem \eqref{app-thm:free-energy} is equal to $F$ given in Eq. \eqref{eq:free-energy}.
In fact, $F$ is the natural extension to arbitrary even probability distributions of $\delta J$. 
More specifically, in the case of an arbitrary even probability distribution $\dP[\delta J]=\dP[-\delta J]$, notice that $V_+ G_{\bar{J}+\rms} +V_- G_{\bar{J}-\rms}$  in Eq. \eqref{app-eq:delta-G} is replaced schematically\footnote{
In the general continuum case, one would actually sum over partitions of $\dR$, i.e., $V_{\delta J}$ is the number of lattice sites with disorder values in the range $(\delta J- \Delta/2, \delta  J+\Delta/2)$ where $\delta J= n\times \Delta$ and $n\in \dZ$, and the summation is over all $n\in \dZ$. 
Ultimately, one would take $V\to \infty$ and $\Delta \to 0$ so that the summation approximates an integral.
} with
\begin{equation}
    \sum_{\delta J} V_{\delta J} G_{\bar{J}+\delta J}
\end{equation}
where the summation is over all possible values of $\delta J$, and $V_{\delta J}$ is the number of lattice sites with disorder value $\delta J$.
In the limit $V\to \infty$, for ``well-behaved" probability distributions, we expect $V_{\delta J}/V\to \dP[\delta J]$ and thus the quantity can be replaced by
\begin{equation}
    \sum_{\delta J} V_{\delta J} G_{\bar{J}+\delta J} \sim V \times \dE_{\delta J} G_{J+\delta J}
\end{equation}
And the remainder of the proof is similar.
Admittedly, what ``well-behaved" means rigorously is a small subtlety that we do not delve into, since it is not essential to the techniques developed in the proof.
\section{Rewriting the Free Energy Density}
\label{app:MFT-rewrite}
In the main text, we claimed that $F$ can be rewritten as Eq. \eqref{eq:free-energy-rewrite} using the new variables $a,\eta, \vphi$. This was proven using the integral representation of the modified Bessel function, i.e.,
\begin{equation}
    I_0 (|z|) = \int_{\mathbb{S}^1} \exp {\Re (z e^{-i\phi})} \frac{d\phi}{2\pi}, \quad z\in \mathbb{C} \cong \mathbb{R}^2
\end{equation}
More specifically, if we use the natural isomorphism and view $w_\pm \in \mathbb{C} \cong \mathbb{R}^2$, then
\begin{align}
    G_J  &= \int_{\phi_+=-\pi}^\pi \int_{\phi_- =-\pi}^\pi \frac{d\phi_\pm}{(2\pi)^2} \exp{ \left[\beta \Re \left(w_+e^{-i\phi_+} +w_- e^{-i\phi_-} \right) \right]} e^{ \beta J \cos (\phi_+ -\phi_-)} \\
    &= \int_{\phi=-\pi}^\pi \frac{d\phi}{2\pi} \, e^{ \beta J \cos (\phi_+ -\phi_-)} \int_{\phi_- =-\pi}^\pi \frac{d\phi_-}{2\pi} \exp{ \left[\beta \Re \left((w_+e^{-i\phi}  +w_-) e^{-i\phi_-} \right) \right]} , \quad \phi = \phi_+-\phi_- \\
    &= \int_{\phi=-\pi}^\pi  \frac{d\phi}{2\pi} \, e^{ \beta J \cos (\phi_+ -\phi_-)} I_0\left(\beta |w_+e^{-i\phi}  +w_-| \right)
\end{align}
A change of variables using $w_\pm \mapsto a,\eta,\vphi$ will result in Eq. \eqref{eq:free-energy-rewrite}. 
\section{Self-Consistency Equations}
\label{app:MFT-self-consistency}

In this section, we obtain the self-consistency equations for $T_{\mathbb{Z}_2}$ assuming that $\eta^\star=1$ (or equivalently, we only need to consider $F(w_+,w_-)$ with $|w_+|=|w_-|=|w|$). Indeed, similar to the $U(1)$ transition with respect to $a$, the $\mathbb{Z}_2$ transition occurs exactly when the 2\ts{nd} order derivative with respect to $\vphi$ has a sign transition. More specifically, the critical temperature $T_{\mathbb{Z}_2}$ must satisfy
\begin{equation}
    \left.\frac{\partial^2}{\partial \vphi^2 } \right|_{\vphi =0} F =0, \quad  \left.\frac{\partial}{\partial |w| } \right|_{\vphi =0} F =0
\end{equation}
Where the 2\ts{nd} equality is to determine the minimizing $|w^\star|$. Notice that
\begin{equation}
    K_J \equiv e^{G_J}= \iint d\phi_\pm \exp[{\beta |w| (\cos (\phi_+ -\vphi/2) +\cos(\phi_- +\vphi/2) -\beta J\cos (\phi_+-\phi_-)}]
\end{equation}
Where we have set $w_\pm = |w| e^{\pm i\vphi/2}$ (possible since $K_J$ is independent of the average phase). It's then easy to check that the self-consistency equations reduce to Eq. \eqref{eq:T-Z2} in the main text.

\section{Critical Temperatures}
\label{app:MFT-critical-T}

In this section, we derive the series expansion of the critical temperatures $T_{U(1)},T_{\mathbb{Z}_2}$ with respect to $\bar{J} \ll \rms$. For simplicity, we shall set $d\inplane =1$. Indeed, $T_{U(1)}$ is determined by the self-consistency equation
\begin{equation}
    T = 1+\mathbb{E}_J r_1 (\beta J)
\end{equation}
where $r_\nu \equiv I_\nu/I_0$ and $I_\nu$ are the modified Bessel functions.
When $\bar{J}=0$, it's clear that $T_{U(1)}=1$ and thus to obtain the linear approximation, we set $T=1+c_1 \bar{J}+\cdots_{\bar{J}}$ where $\cdots_{\bar{J}}$ denotes higher order terms in $\bar{J}$. Therefore, keeping only linear terms, we find that
\begin{align}
    c_1 \bar{J} +\cdots &= \mathbb{E}_{\delta J} r_1( (\bar{J}+\delta J)(1-c_1 \bar{J}+\cdots_{\bar{J}}))\\
    &=  \mathbb{E}_{\delta J} r_1( \delta J +(1-c_1 \delta J) \bar{J}+\cdots_{\bar{J}}) \\
    &=  \mathbb{E}_{\delta J} r_1 (\delta J) +\bar{J} \times\mathbb{E}_{\delta J} [r_1'(\delta J)(1-c_1 \delta J)] +\cdots_{\bar{J}}\\
    &=   \bar{J} \times\mathbb{E}_{\delta J} r_1'(\delta J) +\cdots_{\bar{J}}
\end{align}
Where the last equality uses the fact that $\delta J \mapsto r_1(\delta J)$ is odd. Therefore,
\begin{align}
    T &= 1+\bar{J}\times \overline{r_1' ( \delta J)} +\cdots_{\bar{J}} \\
    &=1+\bar{J} \times \left[\frac{1}{2} -\frac{1}{2^3} \overline{(\delta J)^2} +\cdots \right] +\cdots \bar{J}
\end{align}
Where the last equality uses the Taylor series expansion of the derivative $r_1'$.
The 2\ts{nd} transition $T_{\mathbb{Z}_2}$ can be determined similarly (along with the corresponding magnetization $w^\star$ at $T_{\mathbb{Z}_2}$). More specifically, we can set
\begin{equation}
    T=1+s_1 \bar{J} +\cdots_{\bar{J}}, \quad |w|^2 = t_1 \bar{J}+\cdots_{\bar{J}}
\end{equation}
Where used the fact that $T_{U(1)}=T_{\mathbb{Z}_2}$ when $\bar{J}=0$ and that $F$ is analytic with respect to $|w|^2$ rather than $|w|$ (the singularity in derivative with respect to $|w|$ occurs at $|w|=0$). Notice that
\begin{align}
    \frac{1}{I_0(\beta J)} K_J &= \frac{1}{I_0(\beta J)} \int I_0 (2\beta |w|\cos (\phi/2)) e^{\beta J \cos \phi} d\phi \\
    &=1 + \frac{1}{2} (\beta |w|)^2 (1+r_1) +\cdots_{\beta|w|}
\end{align}
Where $\cdots_{\beta|w|}$ denotes higher order terms in $\beta |w|$ and $r_\nu =r_\nu(\beta J)$ as in the main text. Similarly, notice that
\begin{align}
    \frac{1}{I_0(\beta J)} K_J \times 2\bra \cos \phi_+ \ket_J &=\frac{1}{I_0(\beta J)} \int 2 \cos (\phi/2) I_1 (2\beta |w| \cos (\phi/2)) e^{\beta J \cos \phi} d\phi \\
    &= \beta |w| (1+r_1) +\frac{1}{2^2} (\beta |w|)^3 \left( \frac{3}{2} + 2r_1 +\frac{1}{2} r_2\right) +\cdots_{\beta|w|}
\end{align}
And that
\begin{align}
    \frac{1}{I_0(\beta J)} K_J \times \bra (\sin \phi_+ -\sin \phi_-)^2 \ket_J &=\frac{1}{I_0(\beta J)} \int (1-\cos \phi) (I_0 (2\beta |w| \cos (\phi/2))+I_2 (2\beta |w| \cos (\phi/2)) e^{\beta J \cos \phi} d\phi \\
    &= (1-r_1) +\frac{3}{2^3} (\beta|w|)^2 (1-r_2)  +\cdots_{\beta|w|}
\end{align}
Therefore,
\begin{align}
    2\bra \cos \phi_+ \ket_J &= \beta |w| (1+r_1) -\frac{1}{2^3}(\beta |w|)^3 ( (1+2r_1)^2 -r_2) +\cdots_{\beta|w|} \\
    \bra (\sin \phi_+ -\sin \phi_-)^2 \ket_J &=  (1-r_1) - \frac{1}{2^3} (\beta |w|)^2 (1 +3 r_2 -4r_1^2) +\cdots_{\beta |w|}
\end{align}
Using the self-consistency equations in Eq. \eqref{eq:T-Z2}, we find that up to linear terms in $\bar{J}$, we have the following independent equations
\begin{align}
    s_1 &= \mathbb{E}_{\delta J} \left[+r_1'(\delta J) -\frac{1}{2^3} t_1 (1+4 r_1^2(\delta J) -r_2(\delta J) \right] \\
    s_1 &= \mathbb{E}_{\delta J} \left[ -r_1'(\delta J) -\frac{1}{2^3} t_1 (1-4 r_1^2(\delta J) +3r_2(\delta J) \right]
\end{align}
Solving the system of equations, we find
\begin{align}
    t_1 &= 2\times  \frac{\overline{r_1' (\delta J)}}{\overline{r_1(\delta J)^2} -\overline{r_2(\delta J)}/2} \\
    s_1 &= -\frac{1}{2^2} \frac{\overline{r_1' (\delta J)} (1+\overline{r_2(\delta J)}}{\overline{r_1(\delta J)^2} -\overline{r_2(\delta J)}/2}
\end{align}
In particular, we find that in the limit of $\rms \to 0^+$,
\begin{equation}
    s_1 = -\frac{2}{3} \frac{1}{\rms^2} +O(1)
\end{equation}
\section{The nonlinear sigma model near $T=0$ and small disorder $\rms\to 0^+$}
\label{app:MFT-nls}

\subsection{Why it Breaks Down at Large $\rms$}
As discussed in the main text, one may naively attempt to obtain near $T=0$ results by using the nonlinear sigma model. More specifically, we would like to compute
\begin{equation}
    e^{G_J} = \iint_{[-\pi,\pi]^2} d\phi_\pm \exp \left[ \beta \sum_l |w_l| \cos \phi_l + \beta J \cos (\vphi +\phi_+-\phi_-) \right]
\end{equation}
by replacing $\cos \phi_l \approx 1-\phi_l^2/2$ and extending the integration limits to $[-\pi,\pi]^2 \mapsto \mathbb{R}^2$ since only small values of $\phi_l$ are expected to contribute to the integral in the limit where $\beta |w_l| \to \infty$ 
(as an example, the asymptotic expansion of the modified Bessel function $I_0(x)$ where $x\to \infty$ utilizes this expansion). 
The limit $\beta |w_l^\star| \to \infty$ is expected to be true since near $T=0$ since $\beta \to \infty$ and the minimizers $w_l^\star$ are related to the magnetiztations of each layer, which are expected to be nonzero for sufficiently low temperatures. 
However, we will find that the expansion is ill-regulated when $|w_l|$ is small and thus we cannot readily find the minimizers $w_l^\star$ using the expansion.

Since ultimately the expansion fails for general disorder strength $\rms$, it is instructive to consider the algebraically simple case $|w_+|=|w_-|=|w|$ and $\vphi=\pi/2$. Notice that these restrictions are sufficient since we will rigorously prove in Appendix \eqref{app:MFT-rigor}, that if $\bar{J}=0$, then the minimizer must satisfy $\eta^\star=1,\vphi^\star=\pm \pi/2$ for all temperatures. 
In this case, up to 2\ts{nd} order in $\phi_l$, we have
\begin{align}
    e^{G_J} &\approx e^{2\beta|w|}\iint_{\mathbb{R}^2} d\phi_\pm \exp \left[ -\frac{1}{2} \beta |w| (\phi_+^2+\phi_-^2) - \beta J (\phi_+-\phi_-) \right] \\
    &= e^{2\beta|w|} e^{\beta J^2/|w|} \frac{\pi}{\beta |w|} \\
    -\frac{1}{\beta} \mathbb{E}_J G_J &\approx -\frac{\mathbb{E}_J J^2}{|w|} -2|w|  +\frac{1}{\beta} \log \frac{\beta |w|}{\pi} \\
    F &\approx -\frac{\mathbb{E}_J J^2}{|w|} -2|w|  +\frac{1}{2d\kappa} |w|^2  +\frac{1}{\beta} \log \frac{\beta |w|}{\pi} 
\end{align}
On the right hand side, there are two terms which $\to -\infty$ in the limit $|w| \to 0^+$, i.e., the first and last terms. 
In the case where $\beta \to \infty$, one can argue that the last term $\to 0$ and should not be considered. 
However, the first term $\propto -1/|w|$ is independent of temperature and thus cannot be excluded. It's then clear that we cannot readily minimize $F$ using the nonlinear sigma expansion.

\subsection{Why it Works at Small $\rms$}
As discussed in the main text, in the limit of small disorder $\rms \to 0^+$, the minimizers $a^*(\rms),\eta^*(\rms)$ are expected to be sufficiently close to their non-disordered counterparts $a^*(\rms=0),\eta^*(\rms=0)$, and thus the divergence at $a\to 0^+$ does not affect us. In this case, we can assume that $|w_+|=|w_-| = |w|=2d\inplane$ and find that
\begin{align}
    e^{G_J} &\approx e^{2\beta|w|} e^{\beta J \cos \vphi} \iint_{\mathbb{R}^2} d\phi_\pm \exp \left[ -\frac{1}{2} \beta \phi^T A\phi  -\beta \phi^T \psi\right] 
\end{align}
Where $\phi = (\phi_+ ,\phi_-)$, $\psi = J \sin \vphi (1,-1)$, and
\begin{eqnarray}
    A = 
    \begin{bmatrix}
        |w| +J\cos \vphi & -J\cos \vphi \\
        -J\cos \vphi & |w| +J\cos \vphi
    \end{bmatrix}
\end{eqnarray}
Therefore,
\begin{align}
    e^{G_J} &\approx e^{2\beta|w|} e^{\beta J \cos \vphi} e^{\frac{1}{2} \psi^T A\psi} \frac{\pi}{\beta\sqrt{\det A}} \\
    \label{eq:nls-small-disorder}
    -\frac{1}{\beta} \mathbb{E}_J G_J &\approx -\sin^2 \vphi \mathbb{E}_J \left[ \frac{J^2}{|w|} \left( 1+ \frac{2J}{|w|} \cos \vphi \right)^{-1}\right] -\bar{J} \cos \vphi -2|w|  +\frac{1}{\beta} \log \frac{\beta \sqrt{\det A}}{\pi}
\end{align}

Therefore, in the limit $\beta \to\infty$ and $\bar{J} \ll \rms \to 0^+$, we see that
\begin{align}
    F &\approx -\sin^2 \vphi \frac{\overline{J^2}}{|w|} +\bar{J}\cos \vphi +\cdots \\
    &=-\frac{\overline{J^2}}{|w|} \left[1+ \left(\frac{\bar{J}|w|}{2\overline{J^2}}\right)^2 -\left(\cos \vphi -\left(\frac{\bar{J}|w|}{2\overline{J^2}}\right)\right)^2\right]
\end{align}
Where $\cdots$ contains higher order terms or terms that are independent of $\vphi$. Therefore, we find that the minimizer $\vphi^\star$ satisfies
\begin{align}
    \cos \vphi^* &= \frac{\bar{J}|w|}{2\overline{J^2}} \sim \frac{\bar{J} d\inplane}{\overline{(\delta J)^2}}, \quad \bar{J}/\rms \to 0
\end{align}

\subsection{Subtleties of The Previous Solution}
Notice that there is a small subtlety regarding the first term in Eq. \eqref{eq:nls-small-disorder}. 
Indeed, for the nonlinear sigma approximation to converge, we require that $\det A >0$ and thus $J$ cannot be ``too" negative, as indicated by the denominator of the first term. 
However, if we were to naively take the disorder average for a distribution which is not compactly supported (i.e., nonzero probability for large negative $J$), the first term would diverge. Therefore, for large negative values of $J$, one would need to use the asymptotically exact formula instead of the nonlinear sigma approximation, and the first term would actually be the disorder average over $J$ which is not ``too negative", e.g..,
\begin{equation}
    -\sin^2 \vphi \mathbb{E}_J \left[ 1\left\{1+ \frac{2J}{|w|} \cos \vphi \ge \frac{1}{2} \right\}\frac{J^2}{|w|} \left( 1+ \frac{2J}{|w|} \cos \vphi \right)^{-1}\right]
\end{equation}
Where $1\{A\}$ is the indicator function, i.e., $=1$ if $A$ is true, and $=0$ otherwise, and the cutoff $1/2$ was arbitrarily chosen.
In the $\rms \to 0^+$ limit, the contributions of large negative $J$ would $\to 0$ and thus doesn't affect our argument (hence, the abuse of notation). 
\section{Rigorous results for general temperatures $T$: TRSB behavior}
\label{app:MFT-rigor}

In this section, we shall prove the following two statements.

\begin{theorem}[Orientation]
    \label{app-thm:Z2-orient}
    Let the probability distribution of $\delta J$ be even and let $a>0$ where $a$ is defined in Eq. \eqref{eq:new-variables}. If $\bar{J} \ge 0$, then $F$ is minimized when $|\vphi^\star| \le \pi/2$. Equivalently, $F$ is minimized when $\bar{J} \cos \vphi^\star \ge 0$.
\end{theorem}

\begin{theorem}[TRSB]
    \label{app-thm:TRSB}
    Let the probability distribution of $\delta J$ be even and let $a>0$ where $a$ is defined in Eq. \eqref{eq:new-variables}. If $\bar{J} = 0$, then $F$ is minimized when $|\vphi^\star| = \pi/2$ and $\eta^\star =1$ where $\vphi, \eta$ are defined in Eq. \eqref{eq:new-variables}.
\end{theorem}

Indeed, the disordered average $\bar{J}$, at least schematically, corresponds to an interaction of the form $-\bar{J} \cos \phi(\r)$ in the Hamiltonian, and thus one would expect that the $\mathbb{Z}_2$ order parameter $\vphi^\star$ would be such that $\bar{J} \cos \vphi^\star \ge 0$ so the Hamiltonian can be minimized. Theorem \eqref{app-thm:Z2-orient} formalizes this concept within the context of mean field theory. 
Similarly, it was previously argued (based on physical reasoning) in Ref. \cite{yuan2023inhomogeneity} that a disordered system would exhibit time-reversal-symmetry breaking (TRSB) behavior near $\bar{J} = 0$.
Theorem \eqref{app-thm:TRSB} shows that this is indeed true in context of mean-field theory, and that the TRSB behavior is independent of the specific form of the disorder (as long as it is even with respect to $\delta J$).

To prove the given statements, let us introduce some useful lemmas as follows.
\begin{lemma}
    \label{app-lem:U1-sign}
    Let $f:\mathbb{S}^1 \to \mathbb{C}$ denote an ``well-behaved'' function (say continuous). Then
    \begin{align}
        \int_{\theta=-\pi}^\pi f(\cos \theta, \sin \theta) d\theta &= \sum_{\xi,\eta =\pm 1} \int_{\theta=0}^{\pi/2}f( \xi\cos \theta, \eta \sin \theta) d\theta \\
        &= \frac{1}{4} \sum_{\xi,\eta =\pm 1} \int_{\theta=-\pi}^{\pi}f( \xi|\cos \theta|, \eta |\sin \theta|) d\theta
    \end{align}
\end{lemma}
\begin{proof}
    Let $\sigma = (\cos \theta, \sin \theta) \in \mathbb{S}^1$ and let $\xi, \eta=\pm 1$ denote the sign of the $x,y$ components of $\sigma$, i.e., that of $\cos\theta, \sin \theta$, respectively, so that $\cos \theta = \xi |\cos \theta|$ and $\sin \theta = \eta |\sin \theta|$. The statement is then clear after restricting the integration limits to $\theta \in (0,\pi/2)$ so that $\cos \theta, \sin \theta \ge 0$.
\end{proof}

Let us also introduce the concept of \textit{conic combinations}, that is, a function $f(x,y,z)$ is a \textit{conic combination} of $x,y,z \in \mathbb{R}$ if $f$ can be written as a summation of terms of the form $x^k y^l z^m$ (with nonnegative coefficients) as $k,l,m \in \mathbb{N}$, i.e.,
\begin{equation}
    f(x,y,z) = \sum_{k,l,m \ge 0} c_{klm} x^k y^l z^m, \quad c_{klm} \ge 0
\end{equation}
And we call $f$ a \textit{conic function} of $x,y,z$.

For example, the modified Bessel function $I_0(x)$ is a conic function of $x$ (when viewing it as a Taylor series). Since the explicit values of the nonnegative coefficients $c_{klm}$ will not matter, we will use $\{x, y,z\}$ to denote the conic function $f(x,y,z)$.
If $f(x,y,z)$ is even/odd under parity $x,y,z\mapsto -x,-y,-z$, we will denote it by $\{\cdots\}^\pm$, respectively. Notice that the product of 2 conic functions is also conic, i.e., if $f(x,y)$ and $g(z)$ is conic, then so is $f(x,y)g(z)$ and thus we can write
\begin{equation}
    \{x,y\} \{z\} = \{x,y,z\}
\end{equation}
We now claim the following.
\begin{lemma}
    \label{app-lem:Z2-sum}
    \begin{equation}
        \sum_{x,y,z,...=\pm 1} \{x, y, z,...\} \ge 0
    \end{equation}
\end{lemma}
\begin{proof}
    Notice that
    \begin{align}
        \sum_{x,y,z,...=\pm 1} x^k y^l z^m\cdots &= \left[\sum_{x=\pm 1} x^k \right]\left[\sum_{y=\pm 1} y^l\right]\left[\sum_{z=\pm 1} z^m \right] \cdots \\
        &\ge 0
    \end{align}
    Where we used the fact that $\sum_{x=\pm 1} x^k = 2$ if $k$ is even and $=0$ if $k$ is odd. Therefore, the lemma follows.
\end{proof}
\subsection{A Digression to the Non-disorder $J_2$ Problem}
With the previous lemmas in mind, we are now ready to prove the given statements in Theorem \eqref{app-thm:Z2-orient} and \eqref{app-thm:TRSB}. 
However, since the proof is relatively convoluted for the disordered problem, we shall first demonstrate its application on the analogous all-to-all model of the non-disordered (inherent) $J_2 \cos 2\phi$ inter-layer interaction, which turns out to be substantially simpler.
More specifically, we shall prove the following.

\begin{theorem}
    \label{app-thm:non-disordered}
    Consider the mean-field Hamiltonian with inherent $J_2>0$ inter-layer coupling, i.e.,
    \begin{align}
        \mathcal{H}_{\mathrm{MF}} &= -\frac{d\inplane}{V} \sum_{\r,\r', l} \sigma_l(\r) \cdot \sigma_l(\r') +\mathcal{H}_{\mathrm{int}} \\
        \mathcal{H}_{\mathrm{int}} &= \sum_{\r} (-J_1 \cos \phi(\r) +J_2 \cos 2\phi(\r)) \\
        &= \sum_{\r} (-J_1 \sigma_+(\r)\cdot \sigma_-(\r) +J_2 (2 (\sigma_+(\r)\cdot \sigma_-(\r))^2-1))
    \end{align}
    Then the free energy density of $\mathcal{H}_{\mathrm{MF}}$ is
    \begin{equation}
        \mathcal{F}_{\mathrm{MF}} \equiv \lim_{V\to \infty} \frac{-1}{\beta V} \log \mathcal{Z}_{V,J_1,J_2} = \min_{w_\pm \in \mathbb{R}^2} F(w_+,w_-)
    \end{equation}
    Where
    \begin{align}
        F &= \frac{1}{4d\inplane}(w_+^2+w_-^2) -\frac{1}{\beta} G_{J_1,J_2} \\
        G_{J_1,J_2}  &= \ln \iint d\sigma_+ d\sigma_- e^{-\beta H_{J_1,J_2}} \\
        -H_{J_1,J_2} &= w_+\cdot \sigma_+ +w_- \cdot \sigma_- +J_1 \sigma_+ \cdot \sigma_- - J_2 (2 (\sigma_+ \cdot \sigma_-)^2 -1)
    \end{align}
    Furthermore, define $\vphi, a, \eta$ as in the main text, i.e. Eq. \eqref{eq:new-variables}, for the disordered problem with respect to $w_\pm \in \mathbb{R}^2$. 
    \begin{enumerate}
        \item (Orientation) If $J_1\ge 0$, then $F$ is minimized when $|\vphi^\star| \le \pi/2$ for all $\beta, a,\eta>0$. Equivalently, $F$ is minimized when $J_1 \cos \vphi^\star \ge 0$.
        \item (TRSB) If $J_1=0$, then the free energy $F$ is minimized by $|\vphi^\star| = \pi/2$ and $\eta^\star =1$ for all $\beta,a >0$.
    \end{enumerate}
\end{theorem}
\begin{proof}
    The proof of the free energy density of $\mathcal{H}_{\mathrm{MF}}$ follows a similar argument as that of Theorem \eqref{app-thm:free-energy}. 
    Therefore, let us focus on the final 2 statement. 
    Indeed, since the statement holds for all $\beta,a>0$, we shall simplify notation in our proof by setting $\beta =2a =1$ (One can of course repeat the proof with general values of $\beta,a>0$). 
    We shall also write $K = e^G$ as the partition function of the two-site Hamiltonian $H$ (where we have omitted the subscripts $J_1,J_2$).

    \begin{enumerate}
        \item Let us consider that case where $J_1\ge 0$ and compare the possibilities of $\vphi$ and $\pi-\vphi$, i.e., the sign of
        \begin{equation}
            K(\vphi) - K(\pi -\vphi)
        \end{equation}
        Notice that 
        \begin{equation}
            K_{J_1,J_2}(\vphi) = \int \frac{d\theta}{2\pi} I_0(\sqrt{1+\eta \cos\theta}) e^{J_1 \cos (\theta+\vphi) -J_2 \cos 2(\theta +\vphi)}
        \end{equation}
        Therefore,
        \begin{equation}
            K_{J_1,J_2} (\pi -\vphi) = K_{-J_1, J_2}(\vphi)
        \end{equation}
        And thus we shall consider
        \begin{align}
            K_{J_1,J_2} (\vphi) -K_{-J_1, J_2}(\vphi) &=2 \int \frac{d\theta}{2\pi} I_0(\sqrt{1+\eta \cos\theta}) e^{-J_2 \cos 2(\theta+\vphi)} \sinh( J_1 \cos(\theta+\vphi)) \\
            &= 2 \int \frac{d\theta}{2\pi} I_0(\sqrt{1+\eta \cos(\theta-\vphi)}) e^{-J_2 \cos 2\theta} \sinh( J_1 \cos\theta)
        \end{align}
        Notice that
        \begin{equation}
            I_0(x) = \sum_{n=0}^\infty \frac{1}{n!^2} \left(\frac{x}{2}\right)^{2n}
        \end{equation}
        Therefore, $I_0(\sqrt{1+\eta \cos (\theta-\vphi)})$ is a conic combination of $1+\eta \cos (\theta-\vphi)$. 
        Furthermore, let $\xi_\theta, \eta_\theta=\pm 1$ denote the sign of the $x,y$ (real, imaginary) components of $e^{i\theta}$, and similarly for $\xi_\phi, \eta_\phi$. 
        Then we have
        \begin{equation}
            1+\eta \cos(\theta-\vphi) = 1+\xi_\theta \xi_\phi \times \eta |\cos \theta \cos \vphi| +\eta_\theta \eta_\vphi \times \eta|\sin \theta \sin \vphi|
        \end{equation}
        In particular, we see that $I_0(\sqrt{1+\eta \cos (\theta-\vphi)})$ is a conic combination of $\xi_\theta \xi_\vphi$ and $\eta_\theta \eta_\vphi$, i.e.,
        \begin{equation}
            I_0(\sqrt{1+\eta \cos (\theta-\vphi)}) = \{\xi_\theta \xi_\vphi, \eta_\theta \eta_\vphi \}
        \end{equation}
        Also notice that
        \begin{align}
            e^{-J_2 \cos 2\theta} \sinh (J_1 \cos \theta) &= \xi_\theta \times  e^{-J_2 (2\cos^2 \theta -1)}  \sinh (J_1 |\cos \theta|) \\
            &= \{ \xi_\theta \}^-
        \end{align}
        Therefore, we have
        \begin{align}
            K_{J_1,J_2} (\vphi) -K_{-J_1, J_2}(\vphi) &= 2 \int \frac{d\theta}{2\pi} \{\xi_\theta \xi_\vphi, \eta_\theta \eta_\vphi \} \{ \xi_\theta \eta_\theta\}^- \\
            &= \frac{1}{4} \sum_{\xi_\theta, \eta_\theta = \pm 1}\int \frac{d\theta}{2\pi} \{\xi_\theta \xi_\vphi, \eta_\theta \eta_\vphi \} \{ \xi_\theta \}^- \\
            &= \frac{1}{4} \sum_{\xi_\theta, \eta_\theta = \pm 1}\int \frac{d\theta}{2\pi} \{\xi_\theta , \eta_\theta  \} \{ \xi_\theta \xi_\vphi \}^- \\
            &= \xi_\vphi  \times  \underbrace{\frac{1}{4} \sum_{\xi_\theta, \eta_\theta = \pm 1}\int \frac{d\theta}{2\pi} \{\xi_\theta , \eta_\theta  \} \{ \xi_\theta  \}^-}_{\ge 0}
        \end{align} 
        Where the 2\ts{nd} equality uses Lemma \eqref{app-lem:U1-sign}, the 3rd equality uses the fact that $\xi_\theta \mapsto \xi_\theta \xi_\vphi$ is a bijective mapping from $\{\pm 1\} \to \{\pm 1\}$ (and similarly $\eta_\theta \mapsto \eta_\theta \eta_\vphi$), and the 4th equality uses Lemma \eqref{app-lem:Z2-sum}. 
        Therefore, $K(\vphi) - K(\pi -\vphi)$ has the same sign as $\xi_\vphi $, and thus the statement regarding orientation follows.
        
        \item Setting $J_1=0$, we find that
        \begin{equation}
            K_{J_2}(\vphi) = \int \frac{d\theta}{2\pi} I_0(\sqrt{1+\eta \cos\theta}) e^{-J_2 \cos 2(\theta +\vphi)}
        \end{equation}
        Therefore, we see that
        \begin{align}
            \frac{\partial}{\partial \vphi} K_{J_2} (\vphi) &= 2J_2 \int \frac{d\theta}{2\pi} I_0(\sqrt{1+\eta \cos\theta}) e^{-J_2 \cos 2(\theta +\vphi)} \sin 2(\theta +\vphi) \\
            &= 2J_2 \int \frac{d\theta}{2\pi} I_0(\sqrt{1+\eta \cos (\theta-\vphi)}) e^{-J_2 \cos 2\theta} \sin (2\theta)
        \end{align}
        As before, we have
        \begin{equation}
            I_0(\sqrt{1+\eta \cos (\theta-\vphi)}) = \{\xi_\theta \xi_\vphi, \eta_\theta \eta_\vphi \}
        \end{equation}
        Also notice that
        \begin{align}
            e^{-J_2 \cos 2\theta} \sin (2\theta) &=2 e^{-J_2 (2\cos^2 \theta -1)} \sin \theta \cos\theta \\
            &= \{ \xi_\theta \eta_\theta\}^-
        \end{align}
        Therefore, we have
        \begin{align}
            \frac{\partial}{\partial \vphi} K_{J_2} (\vphi) &= 2J_2 \int \frac{d\theta}{2\pi} \{\xi_\theta \xi_\vphi, \eta_\theta \eta_\vphi \} \{ \xi_\theta \eta_\theta\}^- \\
            &= \frac{J_2}{2} \sum_{\xi_\theta, \eta_\theta = \pm 1}\int \frac{d\theta}{2\pi} \{\xi_\theta \xi_\vphi, \eta_\theta \eta_\vphi \} \{ \xi_\theta \eta_\theta\}^- \\
            &= \frac{J_2}{2} \sum_{\xi_\theta, \eta_\theta = \pm 1}\int \frac{d\theta}{2\pi} \{\xi_\theta , \eta_\theta  \} \{ \xi_\theta \xi_\vphi \eta_\theta \eta_\vphi\}^- \\
            &= \xi_\vphi \eta_\vphi \times  \underbrace{\frac{J_2}{2} \sum_{\xi_\theta, \eta_\theta = \pm 1}\int \frac{d\theta}{2\pi} \{\xi_\theta , \eta_\theta  \} \{ \xi_\theta  \eta_\theta \}^-}_{\ge 0}
        \end{align} 
        Where the 2\ts{nd} equality uses Lemma \eqref{app-lem:U1-sign}, the 3rd equality uses the fact that $\xi_\theta \mapsto \xi_\theta \xi_\vphi$ is a bijective mapping from $\{\pm 1\} \to \{\pm 1\}$ (and similarly $\eta_\theta \mapsto \eta_\theta \eta_\vphi$), and the 4th equality uses Lemma \eqref{app-lem:Z2-sum}. 
        Therefore, $\partial K_{J_2}(\vphi)/\partial \vphi$ has the same sign as $\xi_\vphi \eta_\vphi$.
        In particular, this implies that $K_{J_2}(\vphi)$ is increasing when $\vphi=0\to \pm \pi/2$ and decreases as $\vphi=\pm \pi/2 \to \pi$, and thus $K_{J_2}$ is maximized (indicating $F$ is minimized) when $|\vphi^\star| = \pi/2$.
        
        Let us now optimize $\eta$. Indeed, if we take $\vphi = \pm \pi/2$, we find that
        \begin{equation}
            K_{J_2}(\eta; \vphi=\pm \pi/2) =  \int \frac{d\theta}{2\pi} I_0(\sqrt{1+\eta \cos\theta}) e^{J_2 \cos 2\theta}
        \end{equation}
        Taking the derivative with respect to $\eta$ and we obtain
        \begin{equation}
            \frac{\partial}{\partial \eta}K_{J_2}(\eta) = \int \frac{d\theta}{2\pi} I_1(\sqrt{1+\eta \cos\theta})\frac{1}{2\sqrt{1+\eta \cos\theta}} \times \cos \theta  e^{J_2 \cos 2\theta}
        \end{equation}
        Notice that
        \begin{equation}
            \frac{I_1(x)}{x} = \sum_{n=0}^\infty \frac{1}{n!(n+1)!} \left( \frac{x}{2}\right)^{2n}
        \end{equation}
        Therefore, we see that  $I_1(\sqrt{1+\eta \cos\theta})/\sqrt{1+\eta \cos\theta}$ is a conic combination of $1+\eta \cos\theta$ (similar to the case of $I_0$), and thus we have
        \begin{equation}
            I_1(\sqrt{1+\eta \cos\theta})\frac{1}{2\sqrt{1+\eta \cos\theta}} = \{ \xi_{\theta}\}
        \end{equation}
        Also notice that 
        \begin{equation}
            \cos \theta  e^{J_2 \cos 2\theta} =\{ \xi_\theta \}^-
        \end{equation}
        Therefore, we have
        \begin{align}
            \frac{\partial}{\partial \eta}K_{J_2}(\eta) = \int \frac{d\theta}{2\pi} \{ \xi_\theta\}\{\xi_\theta\}^- \ge 0
        \end{align}
        Where we applied Lemma \eqref{app-lem:U1-sign} and \eqref{app-lem:Z2-sum} as before. Therefore, we see that $F$ is minimized when $|\vphi^\star| =\pi/2$ and $\eta^\star = 1$.
    \end{enumerate}
\end{proof}
\subsection{Returning to the Disorder Problem}
Let us now prove the statements for the disordered problem. 
Since the statements will hold for all temperatures $\beta >0$ and $a>0$, we shall simplify notation in our proof by setting $\beta = 2a= 1$ (One can of course repeat the proof with general values of $\beta, a>0$). We shall also write $K_J = e^{G_J}$ as the partition function of the two-site Hamiltonian $H_J$ in Eq. \eqref{eq:two-site-H}.
\begin{proof}[Proof of Theorem \eqref{app-thm:Z2-orient}]
    Notice that minimizing $F$ with respect to $\vphi$ is the equivalent to maximizing $\mathbb{E}_J G_J$ with respect to $\vphi$ and thus we will prove the latter.
    It's clear that $K_J (\vphi) = K_J (-\vphi)$. Therefore, we shall fix $|\vphi| \le \pi/2$ and compare $\vphi$ with $\pi -\vphi$. Notice that
    \begin{equation}
        K_J (\pi - \vphi) = K_{-J} (\vphi)
    \end{equation}
    Also notice that
    \begin{equation}
        \mathbb{E}_J G_J = \int_{\delta J >0} (G_{\bar{J}+\delta J} + G_{\bar{J}-\delta J}) \mathbb{P}_{\delta J}
    \end{equation}
    Therefore,
    \begin{equation}
        \mathbb{E}_J G_J (\vphi) -\mathbb{E}_J G_J (\pi -\vphi) = \int_{\delta J >0} (G_{\bar{J}+\delta J}(\vphi) + G_{\bar{J}-\delta J}(\vphi) - G_{-\bar{J}-\delta J}(\vphi) - G_{-\bar{J}+\delta J}(\vphi)) \mathbb{P}_{\delta J}
    \end{equation}
    Therefore, it is sufficient to show that
    \begin{equation}
        G_{\bar{J}+\delta J}(\vphi) + G_{\bar{J}-\delta J}(\vphi) - G_{-\bar{J}-\delta J}(\vphi) - G_{-\bar{J}+\delta J}(\vphi) \ge 0
    \end{equation}
    For simplicity, we will suppress the $\vphi$ notation in $G_J (\vphi)$ since each term is understood to be at $\vphi$. 
    Notice that there are two cases which we will treat independently; (1) $\bar{J} \ge \delta J \ge 0$, and (2) $\delta J \ge \bar{J} \ge 0$.
    \begin{enumerate}[(1)]
        \item In case (1), it is sufficient to show that
        \begin{equation}
            G_p \ge G_{-p}, \quad \forall p \ge 0
        \end{equation}
        Or equivalently, we need to show that
        \begin{equation}
            K_p - K_{-p} = \int \frac{d\theta}{2\pi} I_0 (\sqrt{1+\eta \cos \theta}) \sinh (p \cos(\theta-\vphi)) \ge 0
        \end{equation}
        Let $\xi_\theta, \eta_\theta$ be the signs of the $x,y$ components of $e^{i\theta}$ (i.e., $\cos\theta, \sin \theta$, respectively) as in Lemma. \eqref{app-lem:U1-sign}, and similarly, $\xi_\vphi, \eta_\vphi$ for $\vphi$. By using the Taylor expansion of the modified Bessel function $I_0$, we see that $I_0$ is a conic function of $\xi_\theta$, and thus
        \begin{align}
            K_p - K_{-p} &=  \sum_{\xi_\theta, \eta_\theta = \pm 1} \int d\theta \{\xi_\theta\} \{\xi_\theta \xi_\vphi, \eta_\theta \eta_\vphi\}^-  \\
            &= \xi_\vphi \sum_{\xi_\theta, \eta_\theta = \pm 1} \int d\theta \{\xi_\theta\} \{\xi_\theta, \eta_\theta  \xi_\vphi\eta_\vphi\}^- 
        \end{align}
        Since we are summing over all $\eta_\theta = \pm 1$, we can perform the bijective transformation $\eta_\theta \mapsto \eta_\theta \xi_\vphi \eta_\vphi$, to find that
        \begin{align}
            K_p - K_{-p} &= \xi_\vphi \sum_{\xi_\theta, \eta_\theta = \pm 1} \int d\theta \{\xi_\theta\} \{\xi_\theta, \eta_\theta  \}^- 
        \end{align}
        Since the integrand does not depend on $\xi_\vphi,\eta_\vphi$ (which are not summed over $\pm 1$), by Lemma \eqref{app-lem:Z2-sum}, we see that $K_p - K_{-p}$ has the same sign as $\xi_\vphi$. Since we assumed that $|\vphi| \le \pi/2$, we see that
        \begin{equation}
            K_p - K_{-p} \ge 0
        \end{equation}

        \item In case (2), it is sufficient to show that for all $p \ge q \ge 0$, we have
        \begin{align}
            G_p +G_{-q} \ge G_{-p} +G_q
        \end{align}
        Equivalently, we need to show that
        \begin{align}
            K_p K_{-q} &\ge K_{-p} K_q \\
            (K_p -K_{-p}) (K_q +K_{-q}) &\ge (K_p +K_{-p}) (K_q -K_{-q}) \\
            \frac{K_p -K_{-p}}{K_p+K_{-p}} &\ge \frac{K_q -K_{-q}}{K_q+K_{-q}}
        \end{align}
        Therefore, it's sufficient to show that the following function is increasing for $p \ge 0$.
        \begin{equation}
            p \mapsto \frac{K_p -K_{-p}}{K_p+K_{-p}}
        \end{equation}
        Alternatively, it's sufficient to show that
        \begin{align}
            \frac{\partial}{\partial p} \left(\frac{K_p -K_{-p}}{K_p+K_{-p}} \right) &\ge 0 \\
            \frac{2}{(K_p+K_{-p})^2} \left( \dot{K}_p K_{-p} +K_p \dot{K}_{-p}\right) &\ge 0
        \end{align}
        Where $\dot{K}$ is short for the dertivative of $p\mapsto K_p$ with respect to $p$. Notice that
        \begin{equation}
            \dot{K}_p K_{-p} +K_p \dot{K}_{-p} = \iint I_0(g_1) I_0 (g_2) e^{ p [\cos (\theta_1 +\vphi) -\cos (\theta_1 +\vphi)]} [\cos (\theta_1 +\vphi) +\cos (\theta_1 +\vphi)]
        \end{equation}
        Where $g_i = \sqrt{1+\eta \cos \theta_i}$ and the double integral is over $\theta_1,\theta_2\in (-\pi,\pi)$. Therefore, the \textbf{rest of this proof} will be to show that the $\dot{K}_p K_{-p} +K_p \dot{K}_{-p} \ge 0$ for all $p\ge 0$.
        Notice that by interchanging $\theta_1 \leftrightarrow \theta_2$, we find that
        \begin{equation}
            \dot{K}_p K_{-p} +K_p \dot{K}_{-p} = \iint \mu (\theta -\vphi) f(\theta)
        \end{equation}
        where $\theta= (\theta_1,\theta_2)$ and $\theta-\vphi =(\theta_1-\vphi,\theta_2-\vphi)$ and that
        \begin{align}
            f(\theta) &= [\cos \theta_1 +\cos \theta_2] \cosh \left( p(\cos \theta_1 -\cos \theta_2)\right) \\
            \mu(\theta) &= I_0 (g_1 )I_0 (g_2) \ge 0
        \end{align}
        Let us introduce the coupled variables (change of variables),
        \begin{equation}
            \label{eq:coupled-variables-Z2-orient}
            \chi = \frac{\theta_1+\theta_2}{2}, \quad \delta = \frac{\theta_1-\theta_2}{2}
        \end{equation}
        And let $\xi_\chi, \eta_\chi $ be the sign of the $x,y$ components of $e^{i\chi}$ (i.e., the signs of $\cos \chi, \sin \chi$, respectively), and similarly for $\xi_\delta,\eta_\delta$, as in Lemma. \eqref{app-lem:U1-sign}. Then we have
        \begin{equation}
            f(\theta) = \{ \xi_\chi \xi_\delta \}^- \{ \eta_\chi \eta_\delta\}^+
        \end{equation}
        And that
        \begin{equation}
            \mu (\theta-\vphi) = \sum_{\xi_\psi,\eta_\psi=\pm 1} \int \{\xi_\delta \xi_\chi \xi_\vphi, \xi_\delta \eta_\chi \eta_\vphi, \xi_\chi \eta_\chi \xi_\vphi \eta_\vphi \} d\psi
        \end{equation}
        Where $\xi_\vphi,\eta_\vphi$ are those corresponding to the phase $\vphi$ (we leave the proof of this claim to Lemma \eqref{lem:measure-mu}).
        Combining them together, we find that
        \begin{align}
             \dot{K}_p K_{-p} +K_p \dot{K}_{-p} &= \sum_{\xi_\chi,\eta_\chi, \xi_\delta, \eta_\delta \xi_\psi, \eta_\psi = \pm 1} \iiint d\chi d\delta d\psi \{ \xi_\chi \xi_\delta \}^- \{ \eta_\chi \eta_\delta\}^+ 
             \{\xi_\delta \xi_\chi \xi_\vphi, \xi_\delta \eta_\chi \eta_\vphi, \xi_\chi \eta_\chi \xi_\vphi \eta_\vphi \} 
        \end{align}
        Since we are summing over all $\xi_\chi, \eta_\chi = \pm 1$, we can perform the bijective transformation $\xi_\chi, \eta_\chi \mapsto \xi_\chi \xi_\vphi, \eta_\chi \eta_\vphi$, to find that
        \begin{align}
             \dot{K}_p K_{-p} +K_p \dot{K}_{-p} &= \xi_\vphi \sum_{\xi_\chi,\eta_\chi, \xi_\delta, \eta_\delta \xi_\psi, \eta_\psi = \pm 1} \iiint d\chi d\delta d\psi \{ \xi_\chi \xi_\delta \}^- \{ \eta_\chi \eta_\delta\}^+ \{\xi_\delta \xi_\chi , \xi_\delta \eta_\chi , \xi_\chi \eta_\chi  \}
        \end{align}
        Notice that the integrand inside the summation is now independent of $\xi_\vphi,\eta_\vphi$ (which is not summed over $\pm 1$). Therefore, (using the fact that the product of conic functions is again conic) by Lemma \eqref{app-lem:Z2-sum}, we see that $\dot{K}_p K_{-p} +K_p \dot{K}_{-p}$ has the same sign as $\xi_\vphi$. Since we assumed that $|\vphi| \le \pi/2$, we see that
        \begin{equation}
            \dot{K}_p K_{-p} +K_p \dot{K}_{-p}\ge 0
        \end{equation}
    \end{enumerate}
\end{proof}

\begin{proof}[Proof of Theorem \eqref{app-thm:TRSB}]
Notice that minimizing $F$ with respect to $\vphi,\eta$ is the equivalent to maximizing $\mathbb{E}_J G_J$ with respect to $\vphi,\eta$ and thus we will prove the latter. More specifically,
\begin{enumerate}
    \item We will first show that for any $\eta\in [0,1]$, if $\bar{J} =0$, then the mapping $\vphi \mapsto \mathbb{E}_J G_J$ is an increasing function when $0 \le \vphi \le \pi/2$ and a decreasing function when $\pi/2 \le \vphi \le \pi$, so that $\vphi = \pi/2$ maximizes $\mathbb{E}_J G_J$. Since $G_J(\vphi) = G_J(-\vphi)$, we see that $\vphi = \pm \pi/2$ maximizes $\mathbb{E}_J G_J$.
    \item We will then show that when $\bar{J} =0$ and $\vphi = \pm \pi/2$, the mapping $\eta \mapsto \mathbb{E}_J G_J$ is an increasing function and thus $\eta =1$ maximizes $\mathbb{E}_J G_J$.
\end{enumerate} 

\begin{enumerate}
    \item Notice that
    \begin{equation}
        \mathbb{E}_{\delta J} G_{\delta J} = \int_{\delta J >0} (G_{\delta J} +G_{-\delta J}) \mathbb{P}_{\delta J}
    \end{equation}
    Therefore, it's sufficent to show that $\vphi \mapsto G_{\delta J} +G_{-\delta J} =\log K_{\delta J} K_{-\delta J}$ is increasing when $0 \le \vphi \le \pi/2$ and a decreasing function when $\pi/2 \le \vphi \le \pi$. 
    Indeed, let us consider
    \begin{equation}
        \frac{\partial}{\partial \vphi} K_{\delta J} K_{-\delta J} = (-1) \iint d\theta_1 d\theta_2 \mu(\theta-\vphi) (\sin \theta_1 -\sin \theta_2 ) e^{\delta J (\cos \theta_1 -\cos\theta_2)}
    \end{equation}
    Where $g_i = \sqrt{1+\eta \cos\theta_i}$, the double integral is over $\theta_1,\theta_2\in (-\pi,\pi)$, $\theta =(\theta_1,\theta_2)$ and $\theta-\vphi=(\theta_1-\vphi,\theta_2-\vphi)$ and
    \begin{equation}
        \mu(\theta) = I_0(g_1) I_0(g_2)
    \end{equation}
    By the symmetry regarding interchanging $\theta_1\leftrightarrow \theta_2$, we see that
    \begin{equation}
         \frac{\partial}{\partial \vphi} K_{\delta J} K_{-\delta J} = \iint d\theta_1 d\theta_2 \mu(\theta-\vphi) f(\theta)
    \end{equation}
    Where
    \begin{equation}
        f(\theta) = (-1) (\sin \theta_1 -\sin \theta_2) \sinh (\delta J (\cos \theta_1 -\cos \theta_2))
    \end{equation}
    Let us use the coupled integration variables, i.e.,
    \begin{equation}
        \label{eq:coupled-variables-TRSB}
        \chi = \frac{\theta_1+\theta_2}{2}, \quad \delta = \frac{\theta_1-\theta_2}{2}
    \end{equation}
    And let $\xi_\chi, \eta_\chi $ be the sign of the $x,y$ components of $e^{i\chi}$ (i.e., the signs of $\cos \chi, \sin \chi$, respectively), and similarly for $\xi_\delta,\eta_\delta$, as in Lemma. \eqref{app-lem:U1-sign}. 
    Then we have
    \begin{equation}
        f(\theta) = \{ \eta_\delta \xi_\chi\}^- \{\eta_\delta \eta_\chi\}^- =\{ \xi_\chi\}^- \{\eta_\chi\}^-
    \end{equation}
    And by Lemma \eqref{lem:measure-mu}, we have
    \begin{equation}
        \mu(\theta-\vphi) = \sum_{\xi_\psi,\eta_\psi=\pm 1} \int_{\psi=-\pi}^\pi \{\xi_\delta \xi_\chi \xi_\vphi, \xi_\delta \eta_\chi \eta_\vphi, \xi_\chi \eta_\chi \xi_\vphi \eta_\vphi \} d\psi d\psi
    \end{equation}
    Combining them together, we find that
    \begin{align}
         \frac{\partial}{\partial \vphi} K_{\delta J} K_{-\delta J} &= \sum_{\xi_\chi,\eta_\chi, \xi_\delta, \eta_\delta, \xi_\psi, \eta_\psi = \pm 1} \iiint d\chi d\delta d\psi \{ \xi_\chi \}^- \{ \eta_\chi \}^- \{\xi_\delta \xi_\chi \xi_\vphi, \xi_\delta \eta_\chi \eta_\vphi, \xi_\chi \eta_\chi \xi_\vphi \eta_\vphi \} 
    \end{align}
    Since we are summing over all $\xi_\chi, \eta_\chi = \pm 1$, we can perform the bijective transformation $\xi_\chi, \eta_\chi \mapsto \xi_\chi \xi_\vphi, \eta_\chi \eta_\vphi$, to find that
    \begin{align}
         \frac{\partial}{\partial \vphi} K_{\delta J} K_{-\delta J} &= \xi_\vphi \eta_\vphi \sum_{\xi_\chi,\eta_\chi, \xi_\delta, \eta_\delta, \xi_\psi, \eta_\psi = \pm 1} \iiint d\chi d\delta d\psi \{ \xi_\chi \}^- \{ \eta_\chi \}^-  \{\xi_\delta \xi_\chi , \xi_\delta \eta_\chi , \xi_\chi \eta_\chi \} \nonumber
    \end{align}
    Where we used the fact that $\{x \}^-$ is odd under $x\mapsto -x$.
    Notice that the integrand inside the summation is now independent of $\xi_\vphi,\eta_\vphi$ (which is not summed over $\pm 1$). Therefore, (using the fact that the product of conic functions is again conic) by Lemma \eqref{app-lem:Z2-sum}, we see that $\partial K_{\delta J} K_{-\delta J}/\partial \vphi$ has the same sign as $\xi_\vphi \eta_\vphi$. Therefore, $K_{\delta J} K_{-\delta J}$ is increasing when $0\le \vphi \le \pi/2$ and decreasing when $\pi/2 \le \vphi \le \pi$.

    \item Let us now set $\vphi = \pi/2$ since we now know that it must maximize $\mathbb{E}_{\delta J} G_{\delta J}$ for any $\eta \in [0,1]$, and consider 
    \begin{equation}
        \frac{\partial}{\partial\eta} K_{\delta J}K_{-\delta J} = \iint d\theta_1 d\theta_2 \nu (\theta) e^{\delta J( \sin \theta_1 -\sin \theta_2)}
    \end{equation}
    Where $\theta=(\theta_1,\theta_2)$ and 
    \begin{equation}
        \nu(\theta) = \frac{I_1 (g_1)}{g_1} \cos \theta_1 I_0 (g_2) +I_0(g_1) \frac{I_1 (g_2)}{g_2} \cos \theta_2
    \end{equation}
    Where $g_i =\sqrt{1+\eta \cos \theta_i}$ as before. 
    By using the Taylor expansion of the modified Bessel function, it's clear that
    \begin{equation}
        \frac{I_1 (g_1)}{g_1} = \{\xi_1\}, \quad I_0(g_2) =\{\xi_2\}
    \end{equation}
    Where $\xi_i, \eta_i$ are the signs of $\cos \theta_i, \sin \theta_i$. Hence,
    \begin{equation}
        \nu(\theta) = \{\xi_1,\xi_2\}
    \end{equation}
    By symmetry regarding interchanging $\theta_1\leftrightarrow \theta_2$, we see that
    \begin{align}
        \frac{\partial}{\partial\eta} K_{\delta J}K_{-\delta J} &= \iint \nu(\theta) \cosh ( \delta J(\sin \theta_1 -\sin \theta_2)) \\
        &= \sum_{\xi_1,\eta_1,\xi_2,\eta_2=\pm 1} \iint \{\xi_1,\xi_2\} \{\eta_1, -\eta_2\}^+
    \end{align}
    Notice that we can perform the bijective transform $\eta_2 \mapsto -\eta_2$ and the statement then follows from Lemma \eqref{app-lem:Z2-sum}.
\end{enumerate}
\end{proof}

\begin{lemma}
    \label{lem:measure-mu}
    \begin{equation}
        \mu (\theta-\vphi) = \sum_{\xi_\psi,\eta_\psi=\pm 1} \int_{\psi=-\pi}^\pi \{\xi_\delta \xi_\chi \xi_\vphi, \xi_\delta \eta_\chi \eta_\vphi, \xi_\chi \eta_\chi \xi_\vphi \eta_\vphi \} d\psi
    \end{equation}
    Where $\theta-\vphi =(\theta_1-\vphi,\theta_2-\vphi)$, and $\xi_s, \eta_s$ are the signs of $\cos s, \sin s$ for $s=\chi,\delta,\psi,\vphi$, and $\chi, \delta$ are the coupled variables defined from $\theta_1,\theta_2$ as done in Eq. \eqref{eq:coupled-variables-Z2-orient} and/or \eqref{eq:coupled-variables-TRSB}.
\end{lemma}
\begin{proof}
    Since we wish to compute $\mu(\theta-\vphi)$, let us use the notation $h_i = \sqrt{1+\eta \cos(\theta_i -\vphi)}$. Then we have
    \begin{align}
        \mu(\theta-\vphi) &= I_0 (h_1) I_0 (h_2) \\
        &= \iint e^{h_1 \cos \psi_1 +h_2 \cos \psi_2} \frac{d\psi_1}{2\pi}\frac{d\psi_2}{2\pi} \\
        &= \iint \exp \left[\Re (h_1 +h_2 e^{i\psi} ) e^{i\psi_1}\right] \frac{d\psi_1}{2\pi}\frac{d\psi}{2\pi}, \quad \psi = \psi_2-\psi_1 \\
        &= \int I_0(|h_1 +h_2 e^{i\psi}|) \frac{d\psi}{2\pi} \\
        &= \frac{1}{2} \int \left[ I_0(|h_1 +h_2 e^{i\psi}|) + I_0(|h_1 -h_2 e^{i\psi}|)\right]\frac{d\psi}{2\pi}
    \end{align}
    Notice that 
    \begin{equation}
        I_0(x) = \sum_{n\ge 0} \frac{1}{(n!)^2} \left(\frac{x}{2}\right)^{2n}
    \end{equation}
    And that 
    \begin{align}
        |h_1 \pm h_2 e^{i\psi}|^2 = (h_1^2+h_2^2) \pm 2h_1 h_2 \cos \psi
    \end{align}
    Therefore, we see that
    \begin{align}
        \sum_{s=\pm 1} I_0(|h_1 +s h_2 e^{i\psi}|)  &=\sum_{n\ge 0} \frac{1}{2^{2n} (n!)^2} \sum_{s=\pm 1} |h_1 + s h_2e^{i\psi}|^{2n}
    \end{align}
    For any $n\ge 0$, we see that
    \begin{align}
        \sum_{s=\pm 1} |h_1 + s h_2e^{i\psi}|^{2n} &=\sum_{s=\pm 1} \left[ (h_1^2 +h_2^2) +s (2h_1 h_2 \cos \psi)\right]^n \\
        &= \sum_{m\le n} {n \choose m} (h_1^2 +h_2^2)^{n-m} (2 h_1 h_2 \cos \psi)^m \sum_{s=\pm 1} s^m \\
        &= \sum_{m=2k\le n} \underbrace{{n \choose m} 2^{2k+1}  \cos^{2k} \psi}_{\ge 0}  \times (h_1^2 +h_2^2)^{n-2k} (h_1 h_2)^{2k}
    \end{align}
    Notice that the coefficient in front of $(h_1^2 +h_2^2)^{n-2k} (h_1 h_2)^{2k}$ is nonnegative and thus we only need to consider powers of $h_1^2+h_2^2$ and $(h_1 h_2)^2$. Indeed, notice that
    \begin{align}
        h_1^2 +h_2^2 &= 2(1+\eta \cos(\chi-\vphi) \cos \delta) \\
        &= 2(1+\eta \cos \delta (\cos \chi \cos \vphi +\sin \chi \sin \vphi)) \\
        &= 2(1+\eta |\cos \delta |\xi_\delta (|\cos \chi||\cos \vphi| \xi_\chi \xi_\vphi +|\sin \chi||\sin\vphi| \eta_\chi \eta_\vphi)) \\
        &= 2+\underbrace{2\eta |\cos \delta \cos \chi \cos \vphi|}_{\ge 0} \times \xi_\delta \xi_\chi \xi_\vphi +\underbrace{2\eta |\cos \delta \sin \chi \sin\vphi|}_{\ge 0} \times \xi_\delta \eta_\chi \eta_\vphi\\
        &= \{\xi_\delta \xi_\chi \xi_\vphi, \xi_\delta \eta_\chi \eta_\vphi\}
    \end{align}
    And that
    \begin{align}
        (h_1 h_2)^2 &= 1-\eta^2 +\eta^2 (\cos^2 (\chi -\vphi) +\cos^2 \delta) +2\eta \cos(\chi -\vphi) \cos \delta \\
        &= 1-\eta^2 +\eta^2 (\cos^2 \chi \cos^2 \vphi + \sin^2 \chi \sin^2 \vphi +\cos^2 \delta) +\frac{1}{2} \eta^2 \sin 2\chi \sin 2\vphi \\
        &\qquad \qquad \qquad+2\eta \cos \delta (\cos \chi \cos \vphi +\sin \chi \sin \vphi) \\
        &= \underbrace{1-\eta^2 +\eta^2 (\cos^2 \chi \cos^2 \vphi + \sin^2 \chi \sin^2 \vphi +\cos^2 \delta)}_{\ge 0} +\underbrace{2 \eta^2 |\sin \chi \cos \chi \sin \vphi \cos \vphi|}_{\ge 0} \times  \xi_\chi \eta_\chi \xi_\vphi \eta_\vphi \\
        &\qquad \qquad \qquad +\underbrace{2\eta |\cos \delta \cos \chi \cos \vphi|}_{\ge 0} \times \xi_\delta \xi_\chi \xi_\vphi +\underbrace{2\eta |\cos \delta \sin \chi \sin \vphi|}_{\ge 0} \times \xi_\delta \eta_\chi \eta_\vphi\\
        &= \{\xi_\delta \xi_\chi \xi_\vphi, \xi_\delta \eta_\chi \eta_\vphi, \xi_\chi \eta_\chi \xi_\vphi \eta_\vphi\}
    \end{align}
    Combining this together and using the fact that the product of conic functions is again conic, we see that
    \begin{equation}
        \frac{1}{2} \sum_{s=\pm 1} I_0(|h_1 +s h_2 e^{i\psi}|) = \{\xi_\delta \xi_\chi \xi_\vphi, \xi_\delta \eta_\chi \eta_\vphi, \xi_\chi \eta_\chi \xi_\vphi \eta_\vphi \}
    \end{equation}
    And thus the statement follows.
\end{proof}

\end{document}